\documentclass[11pt]{article}

\usepackage{fullpage}
\usepackage{amsmath,amsthm,amsfonts,dsfont}
\usepackage{amssymb,latexsym,graphicx}
\usepackage{palatino}
\usepackage{mathpazo}
\usepackage{stmaryrd}
\usepackage{mathtools}
\usepackage{hyperref}
\usepackage[lined,boxed]{algorithm2e}
\usepackage{subfigure}
\usepackage{boxedminipage}

\usepackage{geometry}
\newtheorem{theorem}{Theorem}[section]

\newtheorem{proposition}[theorem]{Proposition}
\newtheorem{lemma}[theorem]{Lemma}

\newtheorem{claim}[theorem]{Claim}

\newtheorem{corollary}[theorem]{Corollary}
\newtheorem{definition}[theorem]{Definition}

\newcommand{\N}{\ensuremath{\mathbb{N}}}

\newcommand{\R}{\ensuremath{\mathbb{R}}}
\newcommand{\Z}{\ensuremath{\mathbb{Z}}}

 \newcommand{\eps}{\varepsilon} 
\renewcommand{\epsilon}{\varepsilon}

\renewcommand{\vec}[1]{\ensuremath{\mathbf{#1}}}

\newcommand{\basis}{\ensuremath{\mathbf{B}}}
\newcommand{\problem}[1]{\mbox{#1}\xspace}

\newcommand{\poly}{\mathrm{poly}}

\newcommand{\DGS}[3]{\ensuremath{#1\text{-}\problem{DGS}_{#2}^{#3}}}

\newcommand{\coset}{\ensuremath{\vec{c}}}

\newcommand{\scarequotes}[1]{``#1''}

\newcommand{\M}{\mathcal{M}}

\makeatletter
\def\imod#1{\allowbreak\mkern8mu({\operator@font mod}\,\,#1)}
\makeatother

\newcommand{\lat}{\mathcal{L}}
\newcommand{\gs}[1]{\ensuremath{\widetilde{#1}}}
\DeclareMathOperator{\dist}{dist}
\DeclareMathOperator{\spn}{span}
\DeclarePairedDelimiter\inner{\langle}{\rangle}
\DeclarePairedDelimiter\abs{\lvert}{\rvert}
\DeclarePairedDelimiter\set{\{}{\}}
\DeclarePairedDelimiter\ceil{\lceil}{\rceil}
\DeclarePairedDelimiter\length{\lVert}{\rVert}

\begin{document}

\title{Solving the
Closest Vector Problem in $2^n$ Time---\\
The Discrete Gaussian Strikes Again!} 

\author{
Divesh Aggarwal\thanks{Department of Computer Science, EPFL.}\\
\texttt{Divesh.Aggarwal@epfl.ch}
\and
Daniel Dadush\thanks{Centrum Wiskunde \& Informatica, Amsterdam.}
~\thanks{Funded by NWO project number 613.009.031 in the research cluster DIAMANT.}\\
\texttt{dadush@cwi.nl}
\and
Noah Stephens-Davidowitz\thanks{Courant Institute of Mathematical Sciences, New York
 University.}
~\thanks{This material is based upon work supported by the National Science Foundation under Grant No.~CCF-1320188. Any opinions, findings, and conclusions or recommendations expressed in this material are those of the authors and do not necessarily reflect the views of the National Science Foundation.}\\
\texttt{noahsd@cs.nyu.edu}
}
\date{}
\maketitle

\begin{abstract}
We give a $2^{n+o(n)}$-time and space randomized algorithm for solving the {\em exact}
Closest Vector Problem (\problem{CVP}) on $n$-dimensional Euclidean lattices.  This
improves on the previous fastest algorithm, the deterministic $\widetilde{O}(4^{n})$-time and $\widetilde{O}(2^{n})$-space algorithm of Micciancio and Voulgaris~\cite{MV13}.

We achieve our main result in three steps. First, we show how to modify the sampling algorithm from~\cite{ADRS15} to solve the problem of discrete Gaussian sampling over \emph{lattice shifts}, $\lat - \vec{t}$, with very low parameters. While the actual algorithm is a natural generalization of~\cite{ADRS15}, the analysis uses substantial new ideas. This yields a $2^{n+o(n)}$-time algorithm for approximate \problem{CVP} with the very good approximation factor $\gamma = 1+2^{-o(n/\log n)}$.  
Second, we show that the approximate closest vectors to a target vector $\vec{t}$ can be grouped into \scarequotes{lower-dimensional clusters,} and we use this to obtain a recursive reduction from exact CVP to a variant of approximate CVP that \scarequotes{behaves well with these clusters.} 
Third, we show that our discrete Gaussian sampling algorithm can be used to solve this variant of approximate CVP.

The analysis depends crucially on some new properties of the discrete Gaussian distribution and approximate closest vectors, which might be of independent interest.
\end{abstract}

\textbf{Keywords.}  Discrete Gaussian, Closest Vector Problem, Lattice Problems.

\section{Introduction}
\label{sec:introduction}
A lattice $\lat$ is the set of all integer combinations of
linearly independent vectors $\vec{b}_1,\dots,\vec{b}_n \in \R^n$. The matrix
$\basis=(\vec{b}_1,\dots,\vec{b}_n)$ is called a basis of $\lat$, and we write
$\lat(\basis)$ for the lattice generated by $\basis$.

The two most important computational problems on lattices are the Shortest Vector Problem (\problem{SVP}) and the Closest Vector Problem (\problem{CVP}). Given a basis for a lattice $\lat \subseteq \R^n$,
$\problem{SVP}$ asks us to compute a non-zero vector in $\lat$ of minimal length, and \problem{CVP} asks us to compute a lattice vector nearest in Euclidean distance to a target vector $\vec{t}$.

Starting with the seminal work of~\cite{LLL82}, algorithms for solving these problems either exactly or approximately have been studied intensely. Such algorithms have found applications in factoring polynomials over rationals~\cite{LLL82}, integer programming~\cite{Len83,Kan87,DPV11}, cryptanalysis~\cite{Odl90,JS98,NS01}, checking the solvability
by radicals~\cite{LM83}, and solving low-density subset-sum problems~\cite{CJLOSS92}. More recently, many powerful cryptographic primitives have been constructed whose security is based on the {\em worst-case} hardness of these or related lattice problems~\cite{Ajt96,MR07,Gen09,Reg09,BV11,BLPRS13,BV14}.

In their exact forms, both problems are known to be NP-hard (although SVP is only known to be NP-hard under randomized reductions), and they are even hard to approximate to within a factor of $n^{O(1/\log \log n)}$ under reasonable complexity assumptions~\cite{ABSS93,Ajt98,CN98,BS99,DKRS03,Mic01svp,Khot05svp,HRsvp}. \problem{CVP} is thought to be the \scarequotes{harder} of the two problems, as there is a simple reduction from \problem{SVP} to \problem{CVP} that preserves the dimension $n$ of the lattice \cite{GMSS99}, even in the approximate case, while there is no known reduction in the other direction that preserves the dimension.\footnote{Since both problems are NP-complete, there is necessarily an efficient reduction from \problem{CVP} to \problem{SVP}. However, all known reductions either blow up the approximation factor or the dimension of the lattice by a polynomial factor~\cite{Kan87,DH11}. Since we are interested in an algorithm for solving exact CVP whose running time is exponential in the dimension, such reductions are not useful for us.} Indeed, \problem{CVP} is in some sense nearly \scarequotes{complete for lattice problems,} as there are known dimension-preserving reductions from nearly all important lattice problems to \problem{CVP}, 
such as the Shortest Independent Vector Problem, Subspace Avoidance Problem, Generalized Closest Vector Problem, and the Successive Minima Problem~\cite{Micciancio08}. (The Lattice Isomorphism Problem is an important exception.) None of these problems has a known dimension-preserving reduction to \problem{SVP}.

Exact algorithms for \problem{CVP} and \problem{SVP} have a rich history. Kannan initiated their study with an enumeration-based $n^{O(n)}$-time algorithm for \problem{CVP}~\cite{Kan87}, and many others improved upon his technique to achieve better running times~\cite{Helfrich86,HanrotStehle07,MicciancioWalter15}. Since these algorithms solve \problem{CVP}, they also imply solutions for \problem{SVP} and all of the problems listed above. (Notably, these algorithms use only polynomial space.)

For over a decade, these $n^{O(n)}$-time algorithms remained the state of the art until, in a major breakthrough, Ajtai, Kumar, and Sivakumar (AKS) published the first $2^{O(n)}$-time algorithm for \problem{SVP}~\cite{AKS01}. The AKS algorithm is based on \scarequotes{randomized sieving,} in which many randomly generated lattice vectors are iteratively combined to create successively shorter lattice vectors. The work of AKS led to two major questions: First, can \problem{CVP} be solved in $2^{O(n)}$ time? And second, what is the best achievable constant in the exponent? Much work went into solving both of these problems using AKS's sieving technique~\cite{AKS01,AKS02,NguyenVidick08,AJ08,BN09,PS09,MV10,HPS11}, culminating in a $\widetilde{O}(2^{2.456 n})$-time algorithm for \problem{SVP} and a $2^{O(n)}  (1+1/\eps)^{O(n)}$-time algorithm for $(1+\eps)$-approximate \problem{CVP}.

But, exact \problem{CVP} is a much subtler problem than approximate \problem{CVP} or exact \problem{SVP}. In particular, for any approximation factor $\gamma > 1$, a target vector $\vec{t}$ can have arbitrarily many $\gamma$-approximate closest vectors in the lattice $\lat$. For example, $\lat$ might contain many vectors whose length is arbitrarily shorter than the distance between $\vec{t}$ and the lattice, so that any closest lattice vector is \scarequotes{surrounded by} many $\gamma$-approximate closest vectors. Randomized sieving algorithms for \problem{CVP} effectively sample from a distribution that assigns weight to each lattice vector $\vec{y}$ according to some smooth function of $\length{\vec{y} - \vec{t}}$. Such algorithms face a fundamental barrier in solving exact \problem{CVP}: they can \scarequotes{barely distinguish between} $\gamma$-approximate closest vectors and exact closest vectors for very small $\gamma$. (This problem does not arise when solving \problem{SVP} because upper bounds on the lattice kissing number show that there \emph{cannot} be arbitrarily many $\gamma$-approximate shortest lattice vectors. Indeed, such upper bounds play a crucial role in the analysis of sieving algorithms for exact \problem{SVP}.)

So, the important question of whether \problem{CVP} could be solved exactly in singly exponential time remained open until the landmark algorithm of Micciancio and Voulgaris~\cite{MV13}
(MV), which built upon the approach of Sommer, Feder, and
Shalvi~\cite{SFS09}. MV showed a \emph{deterministic} $\widetilde{O}(4^n)$-time and
$\widetilde{O}(2^n)$-space algorithm for exact \problem{CVP}.  
The MV algorithm uses the
\emph{Voronoi cell} of the lattice---the centrally symmetric polytope
corresponding to the points closer to the origin than to any other lattice
point. Until very recently, this algorithm had the best known asymptotic running
time for \emph{both} \problem{SVP} and \problem{CVP}. Prior to this work, this was
the only known algorithm to solve \problem{CVP} exactly in $2^{O(n)}$
time.

Very recently, Aggarwal, Dadush, Regev, and Stephens-Davidowitz (ADRS) gave a $2^{n + o(n)}$-time and space algorithm for \problem{SVP} \cite{ADRS15}. They accomplished this by giving an algorithm that solves the Discrete Gaussian Sampling problem (\problem{DGS}) over a lattice $\lat$. (As this is the starting point for our work, we describe their techniques in some detail below.) They also showed how to use their techniques to approximate \problem{CVP} to within a factor of $1.97$ in time $2^{n+o(n)}$, but like AKS a decade earlier, they left open a natural question: is there a corresponding algorithm for \emph{exact} \problem{CVP} (or even $(1+o(1))$-approximate \problem{CVP})?

\paragraph{Main contribution. } Our main result is a $2^{n+o(n)}$-time and space algorithm that solves \problem{CVP} exactly via discrete Gaussian sampling. We achieve this in three steps. First, we show how to modify the ADRS sampling algorithm to solve \problem{DGS} over \emph{lattice shifts}, $\lat - \vec{t}$. While the actual algorithm is a trivial generalization of ADRS, the analysis uses substantial new ideas. This result alone immediately gives a $2^{n + o(n)}$-time algorithm to approximate \problem{CVP} to within any approximation factor $\gamma = 1+ 2^{-o(n/\log n)}$.
Second, we show that the approximate closest vectors to a target can be grouped into \scarequotes{lower-dimensional clusters.} We use this to show a reduction from \emph{exact} CVP to a variant of approximate CVP. Third, we show that our sampling algorithm actually solves this variant of approximate CVP, yielding a $2^{n + o(n)}$-time algorithm for \emph{exact} CVP. 

We find this result to be quite surprising as, in spite of much research in this area, all previous \scarequotes{truly randomized} algorithms only gave approximate solutions to \problem{CVP}. Indeed, this barrier seemed inherent, as we described above. Our solution depends crucially on the large number of outputs from our sampling algorithm and new properties of the discrete Gaussian.

\subsection{Our techniques}

\paragraph{The ADRS algorithm for centered \problem{DGS} and our generalization.}
The centered discrete Gaussian distribution over a lattice $\lat$ with parameter $s > 0$, denoted $D_{\lat, s}$, is the probability distribution obtained by assigning to each vector $\vec{y} \in \lat$ a probability proportional to its Gaussian mass, $\rho_s(\lat) := e^{-\pi \length{\vec{y}}^2/s^2}$. As the parameter $s$ becomes smaller, $D_{\lat, s}$ becomes more concentrated on the shorter vectors in the lattice. So, for a properly chosen parameter, a sample from $D_{\lat, s}$ is guaranteed to be a shortest lattice vector with not-too-small probability.

ADRS's primary contribution was an algorithm that solves \problem{DGS} in the
centered case, i.e., an algorithm that samples from $D_{\lat, s}$ for any $s$.
To achieve this, they show how to build a discrete Gaussian
\scarequotes{combiner,} which takes samples from $D_{\lat, s}$ and converts them
to samples from $D_{\lat, s/\sqrt{2}}$. The combiner is based on the simple but powerful
 observation that the average of two vectors sampled from $D_{\lat, s}$ is
distributed exactly as $D_{\lat, s/\sqrt{2}}$, \emph{provided that we condition
on the result being in the lattice} \cite[Lemma 3.4]{ADRS15}. Note that the
average of two lattice vectors is in the lattice if and only if they lie in the same
\emph{coset} of $2\lat$. The ADRS algorithm therefore starts with many samples
from $D_{\lat, s}$ for some very high $s$ (which can be computed efficiently
\cite{Klein00, GPV08, BLPRS13}) and repeatedly takes the average of carefully chosen pairs of
vectors that lie in the same coset of $2\lat$ to obtain samples from the
discrete Gaussian with a much lower parameter.

The ADRS algorithm chooses which vectors to combine via rejection sampling
applied to the cosets of $2\lat$, and a key part of the analysis shows that this
rejection sampling does not \scarequotes{throw out} too many vectors. In
particular, ADRS show that, if a single run of the combiner starts with $M$
samples from $D_{\lat, s}$,  then the output will be $\beta(s) M $ samples from
$D_{\lat, s/\sqrt{2}}$, where the \scarequotes{loss factor} $\beta(s)$ is equal
to the ratio of the \emph{collision probability} of $D_{\lat, s}$ mod $2\lat$
divided by the maximal weight of a single coset (with some smaller
factors that we ignore here for simplicity). It is not hard to check that
for any probability distribution over $2^n$ elements, this loss factor is lower
bounded by $2^{-n/2}$. This observation does not suffice, however, since the combiner must be run many times to solve \problem{SVP}.
It is easy to see that the central coset,
$2\lat$, has maximal weight proportional to $\rho_{s/2}(\lat)$, and ADRS show that
the collision probability is proportional to $\rho_{s/\sqrt{2}}(\lat)^2$.
Indeed, the loss factor for a single step is given by
$\beta(s) = \rho_{s/\sqrt{2}}(\lat)^2/(\rho_s(\lat)\rho_{s/2}(\lat))$.  Therefore, the
\emph{total} loss factor $\beta(s) \beta(s/\sqrt{2}) \cdots
\beta(s/2^{-\ell/2})$ accumulated after running the combiner $\ell$ times is
given by a telescoping product, which is easily bounded by $2^{-n/2}$. So, (ignoring small factors) their sampler returns at least $2^{-n/2}\cdot M$ samples from
$D_{\lat, s/2^{-\ell/2}}$.
The ADRS combiner requires $M \geq
2^{n}$ vectors \scarequotes{just to get started,} so they obtain a $2^{n+
o(n)}$-time algorithm for centered $\problem{DGS}$ that yields $2^{n/2}$
samples.

In this work, we show that some of the above analysis carries over easily to the more general case of shifted discrete Gaussians, $D_{\lat - \vec{t}, s}$ for $\vec{t} \in \R^n$---the distribution that assigns Gaussian weight $\rho_s(\vec{w})$ to each $\vec{w} \in \lat - \vec{t}$. As in the centered case, the average of two vectors sampled from $D_{\lat - \vec{t}, s}$ is distributed exactly as $D_{\lat - \vec{t}, s/\sqrt{2}}$, \emph{provided that we condition on the two vectors landing in the same coset of $2\lat$}. (See Lemma~\ref{lem:sumofgaussians} and Proposition~\ref{prop:combiner}.) We can therefore use essentially the same combiner as ADRS to obtain discrete Gaussian samples from the shifted discrete Gaussian with low parameters. 

The primary technical challenge in this part of our work is to bound the accumulated loss factor $\beta(s) \beta(s/\sqrt{2}) \cdots \beta(s/2^{-\ell/2})$. While the loss factor for a single run of the combiner $\beta(s)$ is again equal to the ratio of the collision probability over the cosets to the maximal weight of a coset, this ratio does not seem to have such a nice representation in the shifted case. (See Corollary~\ref{prop:combiner}.) In particular, it is no longer clear which coset has maximal weight, and this coset can even vary with $s$! To solve this problem, we first introduce a new inequality (Corollary~\ref{cor:RSHolder}), which relates the maximal weight of a coset with parameter $s$ to the maximal weight of a coset with parameter $s/\sqrt{2}$.\footnote{This inequality is closely related to that of \cite{RegevS15}, and it (or the more general Lemma~\ref{lem:RSHolder}) may be of independent interest. Indeed, we use it in two seemingly unrelated contexts in the sequel---to bound the loss factor of the sampler, and to show that cosets that contain a closest vector have relatively high weight.} We then show how to use this inequality to inductively bound the accumulated loss factor by (ignoring small factors)
\begin{equation}
\label{eq:intro_lossfactor}
2^{-n} \cdot \frac{\rho_s(\lat - \vec{t})}{\max_{\vec{c} \in \lat/(2\lat)} \rho_s(\vec{c} - \vec{t})} \geq 2^{-n}
\; .
\end{equation}
So, we only need to start out with $2^{n}$ vectors to guarantee that our sampler will return at least one vector. (Like the ADRS algorithm, our algorithm requires at least $2^n$ vectors \scarequotes{just to get started.})

This is already sufficient to obtain a $2^{n+o(n)}$-time solution to approximate \problem{CVP} for any approximation factor $\gamma = 1+2^{o(-n/\log n)}$. (See Corollary~\ref{cor:approxCVP}.) Below, we show that the loss factor in \eqref{eq:intro_lossfactor} is essentially exactly what we need to construct our \emph{exact} CVP algorithm. In particular, we note that if we start with $T \cdot 2^n$ vectors, then the number of output samples is 
\begin{equation}
\label{eq:intro_lossfactor2}
T \cdot \frac{\rho_s(\lat - \vec{t})}{\max_{\vec{c} \in \lat/(2\lat)} \rho_s(\vec{c} - \vec{t})} = \frac{T}{\max_{\vec{c} \in \lat/(2\lat)}\Pr[D_{\lat - \vec{t}, s} \in \vec{c} - \vec{t}]}
\; .
\end{equation}
I.e., we obtain roughly enough samples to \scarequotes{see each coset whose mass is within a factor $T$ of the maximum.}

\paragraph{A reduction from exact CVP to a variant of approximate CVP. } 
In order to solve \emph{exact} CVP, we consider a new variant of approximate CVP called the \emph{cluster} Closest Vector Problem (cCVP). The goal of cCVP is to find a vector that is not only very close to the target, but also very close to an exact CVP solution. More specifically, a vector $\vec{y} \in \lat$ is a valid solution to $\alpha\text{-}\problem{cCVP}$ if there exists an \emph{exact} closest vector $\vec{y}'$ such that $\length{\vec{y} - \vec{y}'} \leq \alpha \cdot \dist(\vec{t}, \lat)$.  We will show below that approximate closest lattice vectors can be grouped into \scarequotes{clusters} contained in balls of radius $\alpha \cdot \dist(\vec{t}, \lat)$. If $\alpha$ is sufficiently small (i.e., $\alpha \leq C/\sqrt{n}$), then we can find a lower-rank sublattice $\lat' \subset \lat$ so that each cluster is actually contained in a shift of $\lat'$. (I.e., each cluster is contained in a lower-dimensional affine subspace. See Figure~\ref{fig:clusters} for an illustration of the clustering phenomenon.) Furthermore, a cCVP oracle is sufficient to find this sublattice $\lat'$. So, we can solve \emph{exact} CVP by (1) computing $\lat'$; (2) solving $\alpha\text{-}\problem{cCVP}$ to find a lattice vector $\vec{y}$ that is in the \scarequotes{correct} shift of $\lat'$; and then (3) solving CVP recursively over the lower-rank shifted lattice $\lat' + \vec{y}$. (See Claim~\ref{clm:simplenCVPreduction} for the full reduction.)

\begin{figure}[ht]
\begin{center}
\includegraphics[width= 0.85 \textwidth]{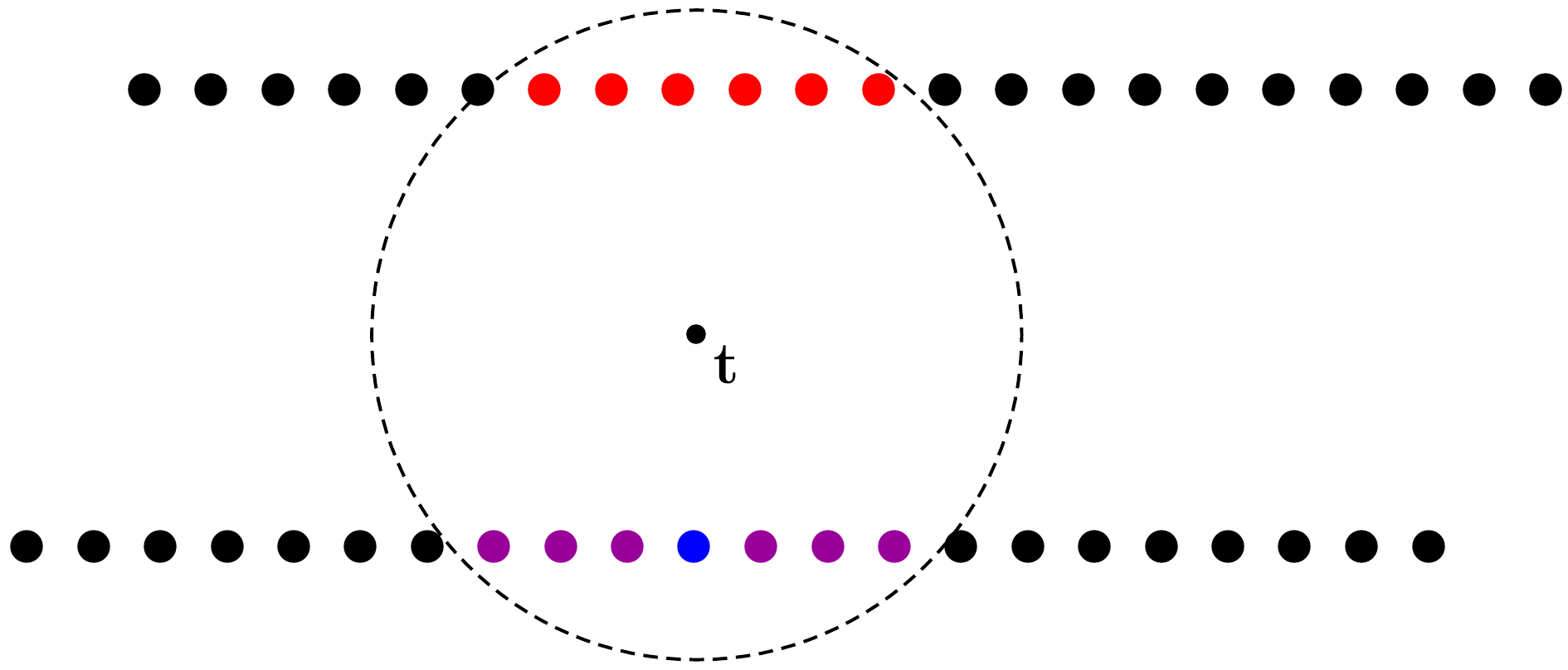}
\caption{\label{fig:clusters} 
A two-dimensional lattice and a target point $\vec{t}$, showing the \scarequotes{clustering} of the approximate closest points. The lattice points inside the dotted circle are approximate closest vectors, and they are clearly organized into two clusters that lie in two distinct one-dimensional affine subspaces. The closest lattice point is highlighted in blue; the points in the same cluster (i.e., the valid solutions to cCVP) are shown in purple; and approximate closest points in a different cluster are shown in red. Notice that close points in the same coset mod $2\lat$ (i.e., points separated by a vector in $2\lat$) are necessarily in the same cluster.}
\end{center}
\end{figure}

This reduction might seem a bit too simple, and indeed we do not know how to use it directly. While we will be able to show that our sampling algorithm does in fact output a solution to cCVP with sufficiently high probability, it will typically output very many vectors, many of which will not be valid solutions to cCVP! We do not know of any efficient way of \scarequotes{picking out} a solution to cCVP from a list of lattice vectors that contains at least one solution. (Note that this issue does not arise for CVP or even approximate CVP, since for these problems we can just take the vector in the list that is closest to the target.) So, we consider an easier problem, $\alpha\text{-}\problem{cCVP}^p$. A valid solution to this problem is a list of at most $p$ lattice vectors, at least one of which lies in the same \scarequotes{cluster} as an exact closest vector, as described above.  (See Definition~\ref{def:nCVP-new}.) This leads to a natural generalization of the reduction described above, as follows. (1) Compute the lower-rank sublattice $\lat' \subset \lat$ as before; (2) solve $\alpha\text{-}\problem{cCVP}^p$ to obtain a list of vectors $(\vec{y}_i, \ldots, \vec{y}_p)$, one of which must lie in the \scarequotes{correct} shift of $\lat'$; (3) solve CVP recursively on all distinct shifts $\lat' + \vec{y}_i$; and finally (4) output the closest resulting point to the target $\vec{t}$

Correctness of this procedure follows immediately from the correctness in the special case when $p=1$. However, bounding the number of recursive calls is more difficult. We accomplish this by first showing that any two of approximate closest vectors $\vec{y}_i, \vec{y}_j$ that are in the same coset mod $2\lat$ must also be in the same cluster. (See Lemma~\ref{lem:clusters}.) This shows that there are at most $2^n$ clusters and therefore at most $2^n$ recursive calls, which would show that the running time is at most roughly $2^{n^2}$. We obtain a much better bound via a technical lemma, which shows that we can always choose the parameters such that either (1) the number of clusters is at most $2^{n-d}$, where $d$ is the rank of the sublattice $\lat'$; or (2) there are \scarequotes{slightly more} than $2^{n-d}$ clusters, but the rank $d$ of $\lat'$ is \scarequotes{significantly less than} $n$. (See Lemma~\ref{lem:good-index}.) This will allow us to show that the total number of calls made on sublattices of rank $d$ after a full run of the algorithm is at most $2^{n-d + o(n)}$. (See Theorem~\ref{thm:CVPtonCVP}.) In particular, this shows that, in order to solve \emph{exact} CVP in time $2^{n+o(n)}$, it suffices to find an algorithm that solves $\alpha\text{-}\problem{cCVP}^p$ for small $\alpha$ that itself runs in time $2^{d+o(d)}$ on lattices of rank $d$.

\paragraph{Solving cluster CVP. } Our final task is to solve $\alpha\text{-}\problem{cCVP}^p$ for sufficiently small $\alpha$ in $2^{n + o(n)}$ time. In other words, we must find an algorithm that outputs a list of approximate closest vectors to the target $\vec{t}$, at least one of which is very close to an exact closest vector. As we noted above, our discrete Gaussian sampler can be used to obtain approximate closest vectors with extremely good approximation factors. Furthermore, any two approximate closest vectors that lie in the same coset mod $2\lat$ must be very close to each other. It therefore suffices to show that at least one of the output vectors of our DGS algorithm will be in the same coset as an exact closest vector mod $2\lat$.

This is why the number of output samples that we computed in~\eqref{eq:intro_lossfactor2} is so remarkably convenient. If a coset's Gaussian mass is within some not-too-large multiplicative factor $T$ of the maximal mass of any coset and we run our sampler, say, $T\cdot \poly(n)$ times, then with high probability one of our output vectors will land in this coset! In particular, if we can find a bound $T \leq 2^{o(n)}$ on the ratio between the maximal mass of any coset and the mass of a coset with a closest vector, then we can simply run our sampler $T \cdot \poly(n)$ times to find a vector in the same coset as this closest vector. In other words, we obtain a $2^{n+o(n)}$-time solution to $\alpha\text{-}\problem{cCVP}^p$, as needed. Intuitively, such a bound $T$ seems reasonable, since a closest vector itself has higher mass than any other point, so one might hope that its coset has relatively high mass.

Unfortunately, we cannot have such a bound for arbitrary $s$. There exist \scarequotes{pathological} lattices $\lat$ and targets $\vec{t}$ such that for some parameter $s$, the coset of a closest vector to $\vec{t}$ has relatively low mass, while some other coset contains many points whose combined mass is quite high, even though it does not contain an exact closest vector. However, we can show that this cannot happen for \scarequotes{too many} different parameters $s$. Specifically, we show how to pick a list of parameters $s_1 \geq \cdots \geq s_\ell$ such that, for at least one of these parameters, the bound $T \leq 2^{o(n)}$ that we required above will hold. This suffices for our purposes. The proof of this statement is quite technical and relies heavily on the new inequality that we prove in Section~\ref{sec:ineq}. (See Corollary~\ref{cor:bigcoset}.)

\subsection{Related work}

Our exact \problem{CVP} algorithm uses many ideas from many different types of lattice algorithms, including sieving, basis reduction, and discrete Gaussian sampling. Our algorithm combines these ideas in a way that (almost magically, and in ways that we do not fully understand) avoids the major pitfalls of each. We summarize the relationship of our algorithm to some prior work below.

First, our algorithm finds an approximate Hermite-Korkine-Zolatoreff (HKZ) basis and essentially \scarequotes{guesses} the last $n-k$ coefficients of a closest vector with respect to this basis.
HKZ bases are extremely well-studied by the basis reduction
community~\cite{Kan87,Helfrich86,lagarias90,HanrotStehle07,MicciancioWalter15},
and this idea is used in essentially all enumeration algorithms for \problem{CVP}. 
However, there are examples where the standard basis enumeration
techniques require $n^{\Omega(n)}$ time to solve CVP. (See, e.g.,~\cite{BGJ14}.) 
The main reason for this is that such 
techniques work recursively on \emph{projections} of the base lattice, and the projected lattice often contains many points close to the
projected target that do not \scarequotes{lift} to points close to the target in the
full lattice. Using our techniques, we never need to project, and we are therefore able to ignore these useless
points while still guaranteeing that we will find a point whose last
 $n-k$ coefficients with respect to the basis are equal
to those of the closest vector.

Many other authors have noted that the approximate closest lattice vectors form clusters, mostly in the context of AKS-like sieving algorithms.
For example, the $(1+\eps)$-approximate closest vectors to $\vec{t}$ can be grouped
into $2^{O(n)}(1+1/\eps)^n$ clusters of diameter $\eps \cdot \dist(\vec{t}, \lat)$ (see, e.g.,~\cite{AJ08,DK13}). While the clustering bound that we
obtain is both stronger and simpler to prove (using an
elementary parity argument), we are unaware of prior work
mentioning this particular bound. This is likely because sieving algorithms
are typically concerned with constant-factor approximations, whereas our sampler allows us to work with
``unconscionably'' good
approximation factors $\gamma = 1+2^{-o(n/\log n)}$. Our clustering bound seems to be both less natural and less useful for the constant-factor approximations achieved by $2^{O(n)}$-time sieving algorithms.

\cite{BonifasD14} improve on the MV algorithm by showing that, once the Voronoi
cell of $\lat$ has been computed, \problem{CVP} on $\lat$ can be solved in
$\widetilde{O}(2^n)$ expected time. Indeed, before we found this algorithm, we 
hoped to solve \problem{CVP} quickly by using the ADRS sampler to compute the Voronoi cell in $2^{n+o(n)}$ time.
(This corresponds to computing
the shortest vectors in every coset of $\lat/(2\lat)$.) Even with our
current techniques, we do not know how to achieve this, and we leave this as an
open problem.

Finally, after this work was published,~\cite{DGStoSVP} showed a dimension-preserving reduction from DGS to CVP, answering a question posed in an earlier version of this paper. Together with our work, this reduction immediately implies a $2^{n+o(n)}$-time algorithm for DGS with \emph{any parameter} $s$. (Our algorithm works for any parameter $s \geq \dist(\vec{t}, \lat) \cdot 2^{o(n/\log n)}$, but not arbitrarily small $s$.) This also provides some (arguably weak) evidence that our technique of using DGS for solving CVP is \scarequotes{correct,} in the sense that any faster algorithm for CVP necessarily yields a faster algorithm for DGS.

\subsection{Open problems and directions for future work}

Of course, the most natural and important open problem is whether a faster algorithm for \problem{CVP} is possible. (Even an algorithm with the same running time as ours that is simpler or deterministic would be very interesting.) There seem to be fundamental barriers to significantly improving our method, as both our sampler and our reduction to exact \problem{CVP} require enumeration over the $2^n$ cosets of $2\lat$. And, Micciancio and Voulgaris note that their techniques also seem incapable of yielding an algorithm that runs in less than $2^n$ time (for similar reasons)~\cite{MV13}. Indeed, our techniques and those of MV seem to inherently solve the harder (though likely not very important) problem of finding \emph{all} closest vectors simultaneously. Since there can be $2^n$ such vectors, this problem trivially cannot be solved in better than $2^n$ time in the worst case. So, if an algorithm with a better running time is to be found, it would likely require substantial new ideas.

Given these barriers, we also ask whether we can find a comparable lower bound. In particular, Micciancio and Voulgaris note that the standard NP-hardness proof for CVP actually shows that, assuming the Exponential Time Hypothesis, there is some constant $c > 0$ such that no $2^{cn}$-time algorithm solves CVP~\cite{MV13}. Recent unpublished work by Samuel Yeom shows that we can take $c = 10^{-4}$ under plausible complexity assumptions \cite{priv:Vinod}. Obviously, this gap is quite wide, and we ask whether we can make significant progress towards closing it.

In this work, we show how to use a technique that seems \scarequotes{inherently approximate} to solve \emph{exact} \problem{CVP}. I.e., our algorithm is randomized and, during any given recursive call, each $\gamma$-approximate closest vector has nearly the same likelihood of appearing as an exact closest vector for sufficiently small $\gamma$. Indeed, prior to this work, the only known algorithm that solved exact \problem{CVP} in $2^{O(n)}$ time was the deterministic MV algorithm, while the \scarequotes{AKS-like} randomized sieving algorithms for \problem{CVP} achieve only constant approximation factors. It would be very interesting to find exact variants of the sieving algorithms. The primary hurdle towards adapting our method to such algorithms seems to be the very good approximation factor that we require---our ideas seem to require an approximation factor of at most $\gamma = 1+1/\poly(n)$, while $2^{O(n)}$-time sieving algorithms only achieve constant approximation factors. But, it is plausible that our techniques could be adapted to work in this setting, potentially yielding an \scarequotes{AKS-like} algorithm for exact CVP. Even if such an algorithm were not provably faster than ours, it might be more efficient in practice, as sieving algorithms tend to outperform their provable running times (while our algorithm quite clearly runs in time at least $2^n$).

A long-standing open problem is to find an algorithm that solves $\problem{CVP}$ in $2^{O(n)}$ time but \emph{polynomial} space. Currently, the only known algorithms that run in polynomial space are the enumeration-based method of Kannan and its variants, which run in $n^{O(n)}$ time. Indeed, even for \problem{SVP}, there is no known polynomial-space algorithm that runs in $2^{O(n)}$ time. This is part of the reason why $n^{O(n)}$-time enumeration-based methods are often used in practice to solve large instances of $\problem{CVP}$ and $\problem{SVP}$, in spite of their much worse asymptotic running time.

The authors are particularly interested in finding a better explanation for why \scarequotes{everything seems to work out} so remarkably well in the analysis of our algorithm. It seems almost magical that we end up with exactly as many samples as we need for our \problem{CVP} to \problem{DGS} reduction to go through. We do not have a good intuitive understanding of why our sampler returns the number of samples that it does, but it seems largely unrelated to the reason that our \problem{CVP} algorithm needs as many samples as it does. The fact that these two numbers are the same is remarkable, and we would love a clear explanation. A better understanding of this would be interesting in its own right, and it could lead to an improved algorithm.

\subsection*{Organization}
In Section~\ref{sec:prelims}, we provide an overview of the
necessary background material and give the basic definitions used throughout
the paper. In Section~\ref{sec:ineq}, we derive an inequality (Corollary~\ref{cor:RSHolder}) that will allow us to bound the \scarequotes{loss factor} of our sampler and the running time of our exact \problem{CVP} algorithm. In Section~\ref{sec:main-dgs}, we present our discrete Gaussian sampler, which immediately yields an approximate CVP algorithm. In Section~\ref{sec:CVPtoneighbor}, we analyze the structure of the approximate closest vectors and show that this leads to a reduction from exact CVP to a variant of approximate CVP. Finally, in Section~\ref{sec:DGSsolvesnCVP}, we show that our DGS algorithm yields a solution to this variant of approximate CVP (and as a consequence, we derive our exact CVP algorithm.)

\section{Preliminaries}
\label{sec:prelims}

Let $\N = \{0,1,2,\ldots \}$. Except where we specify otherwise, we use $C$, $C_1$, and $C_2$ to denote universal positive constants, which might differ from one occurrence to the next (even in the same sequence of (in)equalities). We use bold letters $\vec{x}$ for vectors and denote a vector's coordinates with indices $x_i$. Throughout the paper, $n$ will always be the dimension of the ambient space $\R^n$. 

\subsection{Lattices}

A rank $d$ lattice $\lat\subset \R^n$ is the set of all integer linear combinations of $d$ linearly independent vectors $\basis = (\vec{b}_1, \ldots, \vec{b}_d )$. $\basis$ is called a basis of the lattice and is not unique. Formally, a lattice is represented by a basis $\basis$ for computational purposes, though for simplicity we often do not make this explicit. If $n = d$, we say that the lattice has full rank. We often implicitly assume that the lattice is full rank, as otherwise we can simply work over the subspace spanned by the lattice.

Given a basis, $(\vec{b}_1,\ldots, \vec{b}_d)$, we write $\lat(\vec{b}_1,\ldots, \vec{b}_d)$ to denote the lattice with basis $(\vec{b}_1,\ldots, \vec{b}_d)$. The length of a shortest non-zero vector in the lattice is written $\lambda_1(\lat)$. For a vector $\vec{t} \in \R^n$, we write $\dist(\vec{t}, \lat)$ to denote the distance between $\vec{t}$ and the lattice, $\min_{\vec{y} \in \lat}(\length{\vec{y} - \vec{t}})$. We call any $\vec{y} \in \lat$ minimizing $\length{\vec{y} - \vec{t}}$ a closest vector to $\vec{t}$. 
The covering radius is $\mu(\lat) := \max_{\vec{t}} \dist(\vec{t}, \lat)$.

\begin{definition}
For a lattice $\lat$, the $i$th successive minimum of $\lat$ is
\[ \lambda_i(\lat) = \min \{ r : \dim (\spn (\lat \cap B(\vec0, r))) \geq i \}  \;.\]
\end{definition}

Intuitively, the $i$th successive minimum of $\lat$ is the smallest value $r$ such that there are $i$ linearly independent vectors in $\lat$ of length at most $r$. We will need the following two facts.
\begin{theorem}[{\cite[Theorem 2.1]{BHW93}}]
\label{thm:lat-pt-bnd}
For any lattice $\lat \subset \R^n$ and $s > 0$,
\[
|\set{\vec{y} \in \lat: \|\vec{y}\| \leq s \lambda_1(\lat)}| \leq
2\ceil{2s}^n-1 \text{.} \]
\end{theorem}

\begin{lemma} 
\label{lem:coveringradius}
For any lattice $\lat \subset \R^n$ with basis $(\vec{b}_1,\ldots, \vec{b}_n)$,
\[
\lambda_n(\lat)^2 \leq \mu(\lat)^2 \leq \frac{1}{4} \cdot \sum_{i=1}^n \length{\gs{\vec{b}}_i}^2
\; .
\]
\end{lemma}

\subsection{The discrete Gaussian distribution}

For any $s>0$, we define the function $\rho_s : \R^n \rightarrow\R$ as $\rho_s(\vec{t}) := \exp(-\pi \length{\vec{t}}^2/s^2)$. When $s=1$, we simply write $\rho(\vec{t})$. For a discrete set $A \subset \R^n$ we define $\rho_s(A) :=\sum_{\vec{x}\in A} \rho_s(\vec{x})$. 
\begin{definition} 
For a lattice $\lat \subset \R^n$, a shift $\vec{t} \in \R^n$, and parameter $s > 0$, let $D_{\lat - \vec{t},s}$ be the probability distribution over $\lat - \vec{t}$ such that the probability of drawing $\vec{x} \in \lat - \vec{t}$ is proportional to $\rho_{s}(\vec{x})$. We call this the discrete Gaussian distribution over $\lat - \vec{t}$ with parameter $s$.
\end{definition}

We make frequent use of the discrete Gaussian over the cosets of a sublattice. If $\lat' \subseteq \lat$ is a sublattice of $\lat$, then the set of cosets, $\lat/\lat'$ is the set of translations of $\lat'$ by lattice vectors, $\vec{c} = \lat' + \vec{y}$ for some $\vec{y} \in \lat$. (Note that $\vec{c}$ is a \emph{set}, not a vector.)
Banaszczyk proved the following three bounds~\cite{banaszczyk}.

\begin{lemma}[{\cite[Lemma 1.4]{banaszczyk}}]
\label{lem:banaszczyk} 
For any lattice $\lat\subset\R^n$ and $s > 1$,
\[
\rho_s(\lat) \leq s^n \rho(\lat)
\;.
\]
\end{lemma}

\begin{lemma}
\label{lem:betterrhoLtbound}
For any lattice $\lat\subset\R^n$, $s > 0$, $\vec{t} \in \R^n$
\[
\rho_s(\vec{t}) \leq \frac{\rho_s(\lat - \vec{t})}{\rho_s(\lat)} \leq 1 
\; .
\]
\end{lemma}

\begin{lemma}[{\cite[Lemma 2.13]{cvpp}}]
\label{lem:banaszczyktail} 
For any lattice $\lat\subset\R^n$, $s > 0$, $\vec{t} \in \R^n$, and $r \geq1/\sqrt{2\pi}$,
\[
\Pr_{\vec{X} \sim D_{\lat - \vec{t}, s}}[\length{\vec{X}} \geq r s\sqrt{n} ] < \frac{\rho_s(\lat)}{\rho_s(\lat - \vec{t})}\big( \sqrt{2 \pi e r^2} \exp(-\pi r^2) \big)^n
\; .
\]

\end{lemma}

From these, we derive the following corollary.

\begin{corollary}
\label{cor:banaszczykcor}
For any lattice $\lat \subset \R^n$, $s > 0$, and $\vec{t} \in \R^n$, let $\alpha := \dist(\vec{t}, \lat)/(\sqrt{n}s)$. Then, for any $r \geq 1/\sqrt{2\pi}$,
\begin{equation}
\label{eq:approx-CVP}
\Pr_{\vec{X} \sim D_{\lat - \vec{t}, s}}[\length{\vec{X}} \geq r s\sqrt{n} ] <
e^{\pi n \alpha^2}\big( \sqrt{2 \pi e r^2} \exp(-\pi r^2) \big)^n
\; .
\end{equation}
Furthermore, if $\alpha \leq 2^n$, we have that 
\[
\Pr[\|\vec{X}\|^2 \geq \dist(\vec{t}, \lat)^2 + 2(sn)^2] \leq e^{-3n^2} \text{.}
\]
\end{corollary}
\begin{proof}
We can assume without loss of generality that $\vec{0}$ is a closest vector to
$\vec{t}$ in $\lat$ and therefore $d := \dist(\vec{t}, \lat) = \length{\vec{t}}$.
Equation~\ref{eq:approx-CVP} then follows from combining
Lemma~\ref{lem:betterrhoLtbound} with Lemma~\ref{lem:banaszczyktail}. 

Let $r = \sqrt{\alpha^2 + 2n} \geq 1/\sqrt{2\pi}$, and note that $r s\sqrt{n} =
\sqrt{d^2 + 2(ns)^2}$. Then, by the first part of the corollary, we have that 
\begin{align*}
\Pr[\|\vec{X}\|^2 \geq d^2 + 2(sn)^2] &= \Pr[\|\vec{X}\| \geq rs\sqrt{n}] \\
&\leq e^{\pi \alpha^2 n} \cdot \left(2 \pi e(\alpha^2+2n) \right)^{n/2} \cdot
e^{-n\pi (\alpha^2+2n)} \\
&\leq \left(4 \pi e 2^{2n}\right)^{n/2} e^{-2\pi n^2}  \\
&\leq e^{(\ln(4 \pi e)/2) n + (\ln 2) n^2 - 2\pi n^2} \\
&\leq e^{-3n^2} \; ,
\end{align*}
as needed.
\end{proof}

\subsection{The Gram-Schmidt orthogonalization and \texorpdfstring{$\gamma$}{gamma}-HKZ bases}

Given a basis, $\basis = (\vec{b}_1,\ldots, \vec{b}_n)$, we define its Gram-Schmidt orthogonalization $(\gs{\vec{b}}_1,\ldots, \gs{\vec{b}}_n)$ by
\[  \gs{\vec{b}}_i = \pi_{\{b_1, \ldots, b_{i-1} \}^\perp}(\vec{b}_i) \; , \]
and the corresponding Gram-Schmidt coefficients $\mu_{i,j}$ by
\[ \mu_{i,j}= \frac{\inner{\vec{b}_i, \gs{\vec{b}}_j}}{\length{\gs{\vec{b}}_j}^2}\; . \]
Here, $\pi_A$ is the orthogonal projection on the subspace $A$ and $\{\vec{b}_1,
\ldots, \vec{b}_{i-1} \}^\perp$ denotes the subspace orthogonal to $\vec{b}_1,
\ldots, \vec{b}_{i-1}$.

\begin{definition}
A basis $\basis= (\vec{b}_1,\ldots, \vec{b}_n )$ of $\lat$ is a $\gamma$-approximate Hermite-Korkin-Zolotarev ($\gamma$-HKZ) basis if
\begin{enumerate}
\item $\length{\vec{b}_1} \leq \gamma \cdot \lambda_1(\lat)$;
\item the Gram-Schmidt coefficients of $\basis$ satisfy $\abs{ \mu_{i,j} } \leq \frac{1}{2}$ for all $j<i$; and
\item $\pi_{\{ \vec{b_1} \}^\perp}(\vec{b}_2), \ldots, \pi_{\{ \vec{b}_1 \}^\perp}(\vec{b}_n)$ is a $\gamma$-HKZ basis of $\pi_{\{ \vec{b_1} \}^\perp}(\lat)$.
\end{enumerate}
\end{definition}

We use $\gamma$-HKZ bases in the sequel to find \scarequotes{sublattices that contain all short vectors.} In particular, note that if $(\vec{b}_1, \ldots, \vec{b}_n)$ is a $\gamma$-HKZ basis for $\lat$, then for any index $k$, $\lat(\vec{b}_1,\ldots, \vec{b}_{k-1})$ contains all lattice vectors $\vec{y} \in \lat$ with $\length{\vec{y}} < \length{\gs{\vec{b}}_{k}}/\gamma$. When $\gamma = 1$, we omit it.

\subsection{Lattice problems}

\begin{definition}
For $\gamma = \gamma(n) \geq 1$ (the approximation factor), the search problem $\gamma\text{-}\problem{CVP}$ (Closest Vector Problem) is defined as follows: The input is a basis $\basis$ for a lattice $\lat \subset \R^n$ and a target vector $\vec{t} \in \R^n$. The goal is to output a vector $\vec{y} \in \lat $ with $\length{\vec{y} - \vec{t}} \leq \gamma \cdot \dist(\vec{t}, \lat)$.
\end{definition}

When $\gamma = 1$, we omit it and call the problem \emph{exact} \problem{CVP} or simply \problem{CVP}.

\begin{definition}
\label{def:dgs}
For $\eps \geq 0$ (the error), $\sigma$ (the minimal parameter) a function that maps shifted lattices to non-negative real numbers,
and $m$ (the desired number of output vectors) a function that maps shifted lattices and positive real numbers to natural numbers,
$\DGS{\eps}{\sigma}{m}$ (the Discrete Gaussian Sampling problem) is defined as follows: 
The input is a basis $\basis$ for a lattice $\lat \subset \R^n$, a shift $\vec{t} \in \R^n$, and a parameter $s > \sigma(\lat - \vec{t})$. The goal is to output a sequence of $\hat{m} \geq m(\lat - \vec{t} ,s)$ vectors whose joint distribution is $\eps$-close to $D_{\lat-\vec{t}, s}^{\hat{m}}$.
\end{definition}

We stress that $\eps$ bounds the statistical distance between the \emph{joint} distribution of the output vectors and $\hat{m}$ independent samples from $D_{\lat - \vec{t},s}$. 

\subsection{Some known algorithms}

The following theorem was proven by Ajtai, Kumar, and Sivakumar~\cite{AKS01}, building on work of Schnorr~\cite{Schnorr87}.
\begin{theorem}
\label{thm:BKZ}
There is an algorithm that takes as input a lattice $\lat \subset \R^n$, target $\vec{t} \in \R^n$, and parameter $u \geq 2$ and outputs a $\gamma$-HKZ basis of $\lat$ and a $\gamma'$-approximate closest vector to $\vec{t}$ in time $2^{O(u)} \cdot \poly(n)$, where $\gamma := u^{n/u}$ and $\gamma' := \sqrt{n} u^{n/u}$.
\end{theorem}

The next theorem was proven by~\cite{GMSS99}.

\begin{theorem}
\label{thm:HKZtoCVP}
For any $\gamma = \gamma(n) \geq 1$, there is an efficient dimension-preserving reduction from the problem of computing a $\gamma$-HKZ basis to $\gamma$-CVP.
\end{theorem}

We will also need the following algorithm.
\begin{theorem}[{\cite[Theorem 3.3]{ADRS15}}]
\label{thm:squaresampler}
There is an algorithm that takes as input $\kappa \geq 2$ (the confidence parameter)
and $M$ elements from $\{1, \ldots, N \}$ and outputs a sequence of elements from the same set such that
\begin{enumerate}
\item \label{item:squareruntime} the running time is $M \cdot \poly(\log \kappa, \log N)$;
\item \label{item:squareinputoutput} each $i \in \{1,\ldots, N\}$ appears at least twice as often in the input as in the output; and
\item  \label{item:squaredistribution} if the input consists of 
$M \geq 10 \kappa^2/\max p_i$ 
independent samples from the distribution that assigns probability $p_i$ to element $i$, then
the output is within statistical distance $C_1 M N \log N \exp(-C_2\kappa)$ of $m$ independent samples with respective probabilities $p_i^2/\sum p_j^2$ where $m \geq  M\cdot \sum p_i^2/(32\kappa \max p_i)$ is a random variable.
\end{enumerate}
\end{theorem}

\section{Some inequalities concerning Gaussians on shifted lattices}
\label{sec:ineq}
We first prove an inequality (Corollary~\ref{cor:RSHolder}) concerning the Gaussian measure over shifted lattices. We will use this inequality to show that our sampler outputs sufficiently many samples; and to show that our recursive \problem{CVP} algorithm will \scarequotes{find a cluster with a closest point} with high probability. 
The inequality is similar in flavor to the main inequality in~\cite{RegevS15}, and it (or the more general form given in Lemma~\ref{lem:RSHolder}) may have additional applications. The proof uses the following identity from~\cite{RegevS15}.

\begin{lemma}[{\cite[Eq. (3)]{RegevS15}}]
\label{lem:RS15} 
For any lattice $\lat\subset\R^n$, any two vectors $\vec{x}, \vec{y} \in \R^n$, and $s > 0$, we have
\[
\rho_s(\lat - \vec{x}) \rho_s(\lat - \vec{y}) = \sum_{\vec{c} \in \lat / (2\lat)} \rho_{\sqrt{2}s}(\vec{c} - \vec{x} - \vec{y}) \rho_{\sqrt{2}s}(\vec{c} - \vec{x} + \vec{y}) %
\; .
\]
\end{lemma}

Our inequality then follows easily.

\begin{lemma}
\label{lem:RSHolder}
For any lattice $\lat\subset\R^n$, any two vectors $\vec{x}, \vec{y} \in \R^n$, and $s > 0$, we have
\[
\rho_s(\lat - \vec{x}) \rho_s(\lat - \vec{y}) \leq \max_{\vec{c} \in \lat/(2\lat)} \rho_{\sqrt{2} s}(\vec{c} - \vec{x} - \vec{y}) \cdot \rho_{\sqrt{2}s}(\lat - \vec{x} + \vec{y}) \; .
\]
\end{lemma}
\begin{proof}
Using Lemma~\ref{lem:RS15}, we get the following.
\begin{align*}
\rho_s(\lat - \vec{x}) \rho_s(\lat - \vec{y}) &= \sum_{\vec{c} \in \lat / (2\lat)} \rho_{\sqrt{2}s}(\vec{c} - \vec{x} - \vec{y}) \rho_{\sqrt{2}s}(\vec{c} - \vec{x} + \vec{y})\\
&\leq \max_{\vec{c} \in \lat/(2\lat)} \rho_{\sqrt{2} s}(\vec{c} - \vec{x} - \vec{y}) \cdot \sum_{\vec{d} \in \lat/(2\lat)} \rho_{\sqrt{2}s}(\vec{d} - \vec{x} + \vec{y}) \\
&= \max_{\vec{c} \in \lat/(2\lat)} \rho_{\sqrt{2} s}(\vec{c} - \vec{x} - \vec{y}) \cdot \rho_{\sqrt{2}s}(\lat - \vec{x} + \vec{y}) \; . \qedhere
\end{align*}
\end{proof}

Setting $\vec{x} = \vec{y} = \vec{w}+\vec{t}$ for any $\vec{w} \in \lat$ and switching $2\lat$ with $\lat$ gives the following inequality.
\begin{corollary}
\label{cor:RSHolder}
For any lattice $\lat\subset\R^n$, $\vec{t} \in \R^n$, and $s > 0$, we have
\[
\max_{\vec{c} \in \lat/(2\lat)} \rho_{s}(\vec{c} - \vec{t})^2 \leq \max_{\vec{c}
\in \lat/(2\lat)} \rho_{s/\sqrt{2}}(\vec{c} - \vec{t}) \cdot \rho_{s/\sqrt{2}}(\lat)
\;.
\]
\end{corollary}

\section{Sampling from the discrete Gaussian}
\label{sec:main-dgs}

\subsection{Combining discrete Gaussian samples}

The following lemma and proposition are the shifted analogues of~\cite[Lemma 3.4]{ADRS15} and~\cite[Proposition 3.5]{ADRS15} respectively. Their proofs are nearly identical to the related proofs in \cite{ADRS15}, and we include them in the appendix for completeness. (We note that Lemma~\ref{lem:sumofgaussians} can be viewed as a special case of Lemma~\ref{lem:RS15}.)

\begin{lemma}
\label{lem:sumofgaussians}
Let $\lat \subset \R^n$, $s > 0$ and $\vec{t} \in \R^n$. Then for all $\vec{y} \in \lat - \vec{t}$,
\begin{align}\label{eq:sumofgaussians}
\Pr_{(\vec{X}_1, \vec{X}_2) \sim D_{\lat-\vec{t}, s}^2}[(\vec{X}_1 + \vec{X}_2)/2 = \vec{y} ~|~ \vec{X}_1 + \vec{X}_2 \in 2\lat - 2\vec{t}] 
= \Pr_{\vec{X} \sim D_{\lat-\vec{t}, s/\sqrt{2}}}[\vec{X} = \vec{y}]
\; .
\end{align}
\end{lemma}

\begin{proposition}
\label{prop:combiner}
There is an algorithm that takes as input a lattice $\lat \subset \R^n$, $\vec{t} \in \R^n$, $\kappa \geq 2$ (the confidence parameter), 
and a sequence of vectors from $\lat - \vec{t}$,
and outputs a sequence of vectors from $\lat - \vec{t}$ such that, if the input consists of 
\[
M \geq 10 \kappa^2\cdot \frac{ \rho_s(\lat - \vec{t})}{\max_{\vec{c} \in \lat/(2\lat)}\rho_s(\vec{c} - \vec{t})}
\]
 independent samples from $D_{\lat - \vec{t}, s}$ for some $s > 0$, then the output is within statistical distance $M \exp(C_1 n-C_2\kappa)$ of $m$ independent samples from $D_{\lat-\vec{t}, s/\sqrt{2}}$ where $m$ is a random variable with
\[
m \geq  M \cdot \frac{1}{32\kappa}\cdot\frac{\rho_{s/\sqrt{2}}(\lat) \cdot \rho_{s/\sqrt{2}}(\lat - \vec{t})}{ \rho_s(\lat - \vec{t})\max_{\vec{c} \in \lat/(2\lat)}\rho_s(\vec{c} - \vec{t})}
\; .
\] 
The running time of the algorithm is at most $M \cdot \poly(n, \log \kappa)$.
\end{proposition}

We will show in Theorem~\ref{thm:pipeline} that by calling the algorithm from Proposition~\ref{prop:combiner} repeatedly, we obtain a general discrete Gaussian combiner.

\begin{theorem}
\label{thm:pipeline}
There is an algorithm that takes as input a lattice $\lat \subset \R^n$, $\ell \in \N$ (the step parameter), 
$\kappa \geq 2$ (the confidence parameter), $\vec{t} \in \R^n$, and $M = (32\kappa)^{\ell+1} \cdot  2^n$ vectors in $\lat$ such that, if the input vectors are distributed as $D_{\lat-\vec{t}, s}$ for some $s > 0$, then the output is a list of vectors whose distribution is within statistical distance $\ell M \exp(C_1 n-C_2\kappa )$ of at least 
\[
m = \frac{ \rho_{2^{-\ell/2} s}(\lat - \vec{t})}{\max_{\vec{c} \in \lat/(2\lat)}\rho_{2^{-\ell/2} s}(\vec{c} - \vec{t})}
\] independent samples from $D_{\lat-\vec{t}, 2^{-\ell/2}s}$.
The algorithm runs in time $\ell M \cdot \poly(n,\log \kappa)$.
\end{theorem}
\begin{proof}
Let $\mathcal{X}_0 = (\vec{X}_1,\ldots, \vec{X}_M)$ be the sequence of input vectors. For $i = 0,\ldots, \ell-1$, the algorithm calls the procedure from Proposition~\ref{prop:combiner} with input $\lat$, $\kappa$, and $\mathcal{X}_i$, receiving an output sequence $\mathcal{X}_{i+1}$ of length $M_{i+1}$. Finally, the algorithm outputs the sequence $\mathcal{X}_\ell$.

The running time is clear. Fix $\lat$, $s$, $\vec{t}$ and $\ell$. Define $\theta(i) := \rho_{2^{-i/2} s}(\lat)$, $\phi(i) := \max_{\vec{c} \in \lat/(2\lat)} \rho_{2^{-i/2}s}(\vec{c} -\vec{t})$, and $\psi(i) := \rho_{2^{-i/2}s}(\lat-\vec{t})$.

 We wish to prove by induction that $\mathcal{X}_i$ is within statistical distance $i M\exp(C_1 n-C_2\kappa )$ of $D_{\lat - \vec{t},2^{-i/2}s}^{M_i}$ with 
\begin{align}
\label{eq:inductionhypo}
M_i   \geq (32\kappa)^{\ell - i + 1} \cdot \frac{\psi(i)}{\phi(i)} 
\; ,
\end{align}
for all $i \ge 1$. This implies that $M_\ell \geq m$ as needed.

Let
\[
L(i) := \frac{\theta(i+1) \psi(i+1)}{\psi(i) \phi(i)} \;,
\]
be the \scarequotes{loss factor} resulting from the $(i+1)$st run of the combiner, ignoring the factor of $32\kappa$.
By Corollary~\ref{cor:RSHolder}, we have
\begin{equation}
\label{eq:loss}
L(i) \ge \frac{\psi(i+1)}{\phi(i+1)} \cdot \frac{\phi(i)}{\psi(i)} \;. %
\end{equation}
By Proposition~\ref{prop:combiner}, up to statistical distance $M\exp(C_1 n-C_2\kappa )$, we have that $\mathcal{X}_1$ has the right distribution with
\begin{align*}
M_{1}   &\geq \frac{1}{32\kappa }\cdot M_0 \cdot L(0) \\
 &\ge  (32 \kappa)^{\ell} \cdot 2^n \cdot \frac{\psi(1)}{\phi(1)} \cdot \frac{\phi(0)}{\psi(0)}  
\; ,
\end{align*}
where we used Eq.~(\ref{eq:loss}) with $i=0$. 
By noting that $\psi(0) \le 2^n \phi(0)$, we see that~\eqref{eq:inductionhypo} holds when $i=1$. 

Suppose that $\mathcal{X}_i$ has the correct distribution and~\eqref{eq:inductionhypo} holds for some $i$ with $0 \leq i < \ell$. In particular, we have that $M_i$ is at least 
$10 \kappa^2 \psi(i)/\phi(i)$. This is precisely the condition
necessary to apply Proposition~\ref{prop:combiner}. So, 
we can apply the proposition and the induction hypothesis and obtain that (up to statistical distance at most $(i+1) M \exp(C_1 n - C_2 \kappa)$),
$\mathcal{X}_{i+1}$ has the correct distribution with
\[
M_{i+1}  \geq \frac{1}{32\kappa }\cdot M_i \cdot L(i)  
\ge (32\kappa)^{\ell - i} \cdot  \frac{\psi(i)}{\phi(i)} \cdot \frac{\phi(i)}{\psi(i)} \cdot \frac{\psi(i+1)}{\phi(i+1)}
= (32\kappa)^{\ell - i} \cdot  \frac{\psi(i+1)}{\phi(i+1)}
\; ,
\]
where in the second inequality we used the induction hypothesis and Eq.~(\ref{eq:loss}).
\end{proof}

\subsection{Initializing the sampler}

In order to use our combiner, we need to start with samples from the discrete Gaussian distribution with some large parameter $\hat{s}$.  For very large parameters, the algorithm introduced by Klein and further analyzed by Gentry, Peikert, and Vaikuntanathan suffices~\cite{Klein00, GPV08}. 
For convenience, we use the following strengthening of their result due to Brakerski et al., which provides exact samples and gives better bounds on the parameter $s$.
\begin{theorem}[{\cite[Lemma 2.3]{BLPRS13}}]
\label{thm:GPV}
There is a probabilistic polynomial-time algorithm that takes as input a basis $\basis $ for a lattice $\lat \subset \R^n$ with $n \geq 2$, a shift $\vec{t} \in \R^n$, and $\hat{s} > C\sqrt{ \log n} \cdot \length{\gs{\basis}}$ and outputs a vector that is distributed exactly as $D_{\lat - \vec{t}, \hat{s}}$, where $\length{\gs{\basis}} := \max \length{\gs{\vec{b}}_i}$.
\end{theorem}

When instantiated with a $\gamma$-HKZ basis, Theorem~\ref{thm:GPV} allows us to sample with parameter $\hat{s} = \gamma \cdot \poly(n) \cdot \lambda_n(\lat)$. After running our combiner $o(n/\log n)$ times, this will allow us to sample with any parameter $s = \gamma \cdot  \lambda_n(\lat)/2^{o(n/\log n)}$. The following proposition and corollary show that we can sample with any parameter $s = \dist(\vec{t}, \lat)/2^{o(n/\log n)}$ by working over a shifted sublattice that will contain all high-mass vectors of the original lattice.
 
 \begin{proposition}
\label{prop:shiftedsublattice}
 There is an algorithm that takes as input a lattice $\lat \subset \R^n $, shift $\vec{t} \in \R^n$, $r > 0$,  and parameter $u \geq 2$, such that if \[
 r \geq u^{n/u} (1+\sqrt{n} u^{n/u}) \cdot \dist(\vec{t}, \lat)
 \; ,
 \] 
 then the output of the algorithm is $\vec{y} \in \lat$ and a basis $\basis'$ of a (possibly trivial) sublattice $\lat' \subseteq \lat $ such that all vectors from $\lat - \vec{t}$ of length at most $r/u^{n/u} - \dist(\vec{t}, \lat)$ are also contained in $\lat' - \vec{y} - \vec{t}$, and $\length{\gs{\basis}'} \leq r$. The algorithm runs in time $\poly(n) \cdot 2^{O(u)}$.
 \end{proposition}
\begin{proof}
 On input a lattice $\lat \subset \R^n$, $\vec{t} \in \R^n$, and $r > 0$, the algorithm behaves as follows. First, it calls the procedure from  Theorem~\ref{thm:BKZ} to compute a $u^{n/u}$-HKZ basis $\basis  = (\vec{b}_1, \ldots, \vec{b}_n)$ of $\lat$. Let $(\gs{\vec{b}}_1, \ldots, \gs{\vec{b}}_n)$ be the corresponding Gram-Schmidt vectors. Let $k \geq 0$ be maximal such that $\length{\gs{\vec{b}}_i} \le r$ for $1 \le i \le k$, and let $\basis' = (\vec{b}_1, \ldots, \vec{b}_k)$.  Let $\pi_k = \pi_{\{ \vec{b}_1, \ldots, \vec{b}_k \}^\perp}$ and $\M = \pi_k(\lat)$. The algorithm then calls the procedure from Theorem~\ref{thm:BKZ} again with the same $s$ and input $\pi_k(\vec{t})$ and $\M$, receiving as output $\vec{x} = \sum_{i=k+1}^n a_i \pi_k(\vec{b}_i)$ where $a_i \in \Z$, a $\sqrt{n} u^{n/u}$-approximate closest vector to $\pi_k(\vec{t})$ in $\M$. Finally, the algorithm returns $\vec{y} = -\sum_{i=k+1}^n a_i \vec{b}_i$ and  $\basis' = (\vec{b}_1, \ldots, \vec{b}_k)$.
 
The running time is clear, as is the fact that $\length{\gs{\basis'}} \leq r$. It remains to prove that $\lat' - \vec{y} - \vec{t}$ contains all sufficiently short vectors in $\lat - \vec{t}$. If $k = n$, then $\lat' = \lat$ and $\vec{y}$ is irrelevant, so we may assume that $k < n$. Note that, since $\basis$ is a $u^{n/u}$-HKZ basis, $\lambda_1(\M) \geq \length{\gs{\vec{b}}_{k+1}}/u^{n/u} > r/u^{n/u}$.  In particular, $\lambda_1(\M)  > (1+\sqrt{n} \cdot u^{n/u})\cdot \dist(\vec{t}, \lat) \geq (1+\sqrt{n} \cdot u^{n/u})\cdot \dist(\pi_k(\vec{t}), \M)$. So, there is a unique closest vector to $\pi_k(\vec{t})$ in $\M$, and by triangle inequality, the next closest vector is at distance greater than $\sqrt{n} \cdot u^{n/u}\dist(\pi_k(\vec{t}), \M)$. Therefore, the call to the subprocedure from Theorem~\ref{thm:BKZ} will output the exact closest vector $\vec{x} \in \M$ to $\pi_k(\vec{t})$.

Let $\vec{w} \in \lat \setminus (\lat' - \vec{y})$ so that $\pi_k(\vec{w}) \neq \pi_k(-\vec{y}) = \vec{x}$. We need to show that $\vec{w} - \vec{t}$ is relatively long. Since $\basis$ is a $s^{n/s}$-HKZ basis, it follows that
\[
\length{\pi_k(\vec{w}) - \vec{x}} \geq \lambda_1(\M) > r/u^{n/u}
\; .
\]
Applying triangle inequality, we have
\begin{align*}
\length{\vec{w} - \vec{t}} \geq \length{\pi_k(\vec{w}) - \pi_k(\vec{t})}
\geq \length{\pi_k(\vec{w}) - \vec{x}} - \length{ \vec{x} - \pi_k(\vec{t})}
> r/u^{n/u}- \dist(\vec{t}, \lat)
\; ,
\end{align*}
as needed.
 \end{proof}

 \begin{corollary}
 \label{cor:start}
 There is an algorithm that takes as input a lattice $\lat \subset \R^n$ with $n \geq 2$, shift $\vec{t} \in \R^n$, $M \in \N$ (the desired number of output vectors), and parameters $u \geq 2$ and 
$\hat{s} > 0$
and outputs $\vec{y} \in \lat$, a (possibly trivial) sublattice $\lat' \subseteq \lat$, and $M$ vectors from $\lat' - \vec{y} - \vec{t}$ such that
if 
\[
\hat{s} \geq C\sqrt{n \log n} \cdot u^{2n/u} \cdot \dist(\vec{t}, \lat)
\; ,
\]
then the output vectors are distributed as $M$ independent samples from $D_{\lat' - \vec{y} - \vec{t}, \hat{s}}$, and $\lat' - \vec{y} - \vec{t}$ contains all vectors in $\lat - \vec{t}$ of length at most $C\hat{s}/(u^{n/u}\sqrt{\log n})$. The algorithm runs in time $\poly(n) \cdot 2^{O(u)} + \poly(n) \cdot M $.
\end{corollary}
\begin{proof}
 The algorithm first calls the procedure from Proposition~\ref{prop:shiftedsublattice} with input $\lat$, $\vec{t}$, and 
\[
r := \frac{C\hat{s}}{\sqrt{\log n}} \geq u^{n/u} (1+\sqrt{n} u^{n/u}) \cdot \dist(\vec{t}, \lat)
\; ,
\] receiving as output $\vec{y} \in \lat$ and a basis $\basis'$ of a sublattice $\lat' \subset \lat $. It then runs the algorithm from Theorem~\ref{thm:GPV} $M$ times with input $\lat'$, $\vec{y} + \vec{t}$, and $\hat{s}$ and outputs the resulting vectors, $\vec{y}$, and $\lat'$.

The running time is clear. By Proposition~\ref{prop:shiftedsublattice}, $\lat' - \vec{y} - \vec{t}$ contains all vectors of length at most $r/u^{n/u} - \dist(\vec{t}, \lat) \geq  C\hat{s}/(u^{n/u}\sqrt{\log n})$ in $\lat - \vec{t}$, and $\length{\gs{\basis}'} \leq r \leq C \hat{s}/\sqrt{ \log n} $. So, it follows from Theorem~\ref{thm:GPV} that the output has the correct distribution.
\end{proof}

\subsection{The sampler}
We are now ready to present our discrete Gaussian sampler.

\begin{theorem}
\label{thm:DGS}
For any efficiently computable function $f(n)\geq n^{\omega(1)}$, let $\sigma$ be the function defined by $\sigma(\lat - \vec{t}) := \dist(\vec{t}, \lat)/f(n)$ for any lattice $\lat \subset\R^n$ and $\vec{t} \in \R^n$. Let 
\[ 
m(\lat - \vec{t}, s) := \frac{\rho_s(\lat - \vec{t})}{\max_{\vec{c} \in \lat/(2\lat)} \rho_s(\vec{c} - \vec{t})}
\; .
\]
Then, there is an algorithm that solves $\DGS{\eps}{\sigma}{m}$ with $\eps(n) := 2^{-Cn^2}$ in time 
$
2^{n+O(\log n \log f(n))}$.
\end{theorem}
\begin{proof}
We assume without loss of generality that $f(n) \geq 2n > 10$. The algorithm behaves as follows on input a lattice $\lat \subset \R^n$, a shift $\vec{t}$, and a parameter $s > \sigma(\lat - \vec{t})$. First, it runs the procedure from Corollary~\ref{cor:start} with input $\lat$, $\vec{t}$, $M := (Cn^2)^{\ell + 2} \cdot 2^n$ with $\ell :=  C\ceil{\log f(n)}$, $u := C n \log n/ \log f(n) + 2$, and 
\[
\hat{s} := 2^\ell s > C\sqrt{n \log n} \cdot u^{2n/u} \cdot \dist(\vec{t}, \lat)
\; .
\]
(Note that $u^{n/u} \leq f(n)^{C}$.) It receives as output $\lat' \subset \R^n$, $\vec{y} \in \lat$, and $(\vec{X}_1, \ldots, \vec{X}_M) \in \lat' - \vec{y} - \vec{t}$. It then runs the procedure from Theorem~\ref{thm:pipeline} \emph{twice}, first with input $\lat'$, $\ell$, $\kappa := Cn^2$, $\vec{t}$, and the first half of the vectors, $(\vec{X}_1, \ldots, \vec{X}_{M/2})$; and next with input $\lat'$, $\ell$, $\kappa$, $\vec{t}$, and the second half of the vectors, $(\vec{X}_{M/2+1}, \ldots, \vec{X}_{M})$. Finally, it outputs the resulting vectors.

The running time follows from the respective running times of the two subprocedures. In particular, the procedure from Corollary~\ref{cor:start} runs in time $\poly(n) \cdot (2^{O(u)} + M) =  n^{O(n/\log f(n))} + 2^{n+O(\log n \log f(n))} =  2^{n+O(\log n \log f(n))}$, and the procedure from Theorem~\ref{thm:pipeline} runs in time $\ell M \cdot \poly(n, \log \kappa) = 2^{n+O(\log n\log f(n))}$.

By Corollary~\ref{cor:start}, the $\vec{X}_i$ are $M$ independent samples from $D_{\lat' - \vec{y} - \vec{t}, \hat{s}}$  and $\lat' - \vec{y} - \vec{t}$ contains all vectors in $\lat - \vec{t}$ of length at most $C\hat{s} /(u^{n/u}\sqrt{ \log n})$. By Theorem~\ref{thm:pipeline}, the output contains at least $2m(\lat' - \vec{t}, s)$ vectors whose distribution is within statistical distance $2^{-Cn^2}$ of independent samples from $D_{\lat' - \vec{y} - \vec{t}, s}$. 

We now show that $D_{\lat' -\vec{y} - \vec{t}, s}$ is statistically close to $D_{\lat - \vec{t}, s}$. Let $d := \dist(\vec{t}, \lat)$ and 
\[
r := \frac{C 2^\ell}{u^{n/u} \sqrt{n \log n}} \geq f(n)^C \geq \frac{1}{\sqrt{2\pi}}
\; .
\] The statistical distance is exactly 
\begin{align*}
\Pr_{\vec{w} \sim D_{\lat - \vec{t}}, s} \big[\vec{w} \notin \lat' - \vec{y} - \vec{t} \big] &< \Pr_{\vec{w} \sim D_{\lat - \vec{t}}, s} \big[\length{\vec{w}} > C\hat{s}/(u^{n/u}\sqrt{\log n}) \big]\\
&= \Pr_{\vec{w} \sim D_{\lat - \vec{t}}, s} \big[\length{\vec{w}} > rs \sqrt{n} \big]\\
&< e^{\pi d^2/s^2} e^{- f(n)^C}\\
&< 2^{-C n^2}
\; ,
\end{align*}
where we have used Corollary~\ref{cor:banaszczykcor}. It follows that the output has the correct size and distribution. In particular, it follows from applying union bound over the output samples that the distribution of the output is within statistical distance $\eps$ of independent samples from $D_{\lat - \vec{t}, s}$, and an easy calculation shows that $2m(\lat' - \vec{t}, s) > m(\lat - \vec{t}, s)$.
\end{proof}

From Theorem~\ref{thm:DGS} and Corollary~\ref{cor:banaszczykcor}, we immediately get a weaker version of our main result, a $2^{n+o(n)}$-time algorithm for $\gamma\text{-}\problem{CVP}$ for any $\gamma = 1+2^{-o(n/\log n)}$. 

\begin{corollary}
\label{cor:approxCVP}
For any efficiently computable function $f(n) \geq n^{\omega(1)}$, there is an algorithm solving $(1+1/f(n))\text{-}\problem{CVP}$ (with high probability) in time 
$
2^{n+O(\log n \log f(n))}
$. In particular, if $f(n) = 2^{o(n/\log n)}$, the algorithm runs in time $2^{n+o(n)}$.
\end{corollary}

\section{Reduction from exact CVP to a variant of approximate CVP}
\label{sec:CVPtoneighbor}
We now introduce a new variant of approximate \problem{CVP} that suggests a recursive algorithm for exact \problem{CVP}. The goal is to find a lattice point $\vec{y}$ that is within some very small distance $\alpha \cdot \dist(\vec{t}, \lat)$ of a closest point $\vec{y}'$ to the target $\vec{t}$. In Section~\ref{sec:clusters}, we show that the approximate closest points are arranged in \scarequotes{clusters,} where $\vec{y}$ and $\vec{y}'$ are in the same cluster. So, we call this problem the \emph{cluster} Closest Vector Problem (cCVP). 

In fact, it will suffice for our purposes to output many lattice vectors $\vec{y}_1,\ldots, \vec{y}_{\hat{p}}$ with the guarantee that at least one of these points is within distance $\alpha \cdot \dist(\vec{t}, \lat)$ to the closest vector.

\begin{definition}
\label{def:nCVP-new}
For $\alpha = \alpha(n) \geq 0$ (the additive error) and $p = p(n) \geq 1$ (a bound on the output size), the search problem $\alpha\text{-}\problem{cCVP}^p$ (cluster Closest Vector Problem) is defined as follows: The input is a basis $\basis$ for a lattice $\lat \subset \R^n$ and a target vector $\vec{t} \in \R^n$. The goal is to output lattice vectors $\vec{y}_1,\ldots, \vec{y}_{\hat{p}} \in \lat $ with $\hat{p} \leq p(n)$ such that there exists an index $j$ and $\vec{y}' \in \lat$ with $\length{\vec{y}' - \vec{t}} = \dist(\vec{t}, \lat)$ and $\length{\vec{y}_j - \vec{y}'} \leq \alpha(n) \cdot \dist(\vec{t}, \lat)$.
\end{definition}

Note that there is a trivial reduction from $(1+\alpha)\text{-}\problem{CVP}$ to $\alpha\text{-}\problem{cCVP}^p$. Furthermore, we may assume without loss of generality that all of the output vectors are solutions to $(1+\alpha)^2\text{-}\problem{CVP}$. (We can simply throw out any vectors $\vec{y}_j$ with $\length{\vec{y}_j - \vec{t}} \geq (1+\alpha) \cdot \min_i \length{\vec{y}_i - \vec{t}}$.) 

We are primarily interested in $\alpha\text{-}\problem{cCVP}^p$ for very large $p$ (e.g., $p = 2^n$), but we first present a simple recursive reduction from \emph{exact} $\problem{CVP}$ to $\alpha\text{-}\problem{cCVP}^1$ for $\alpha(n) \leq C/\sqrt{n}$. Our more general reduction will essentially just run this procedure many times, with each run corresponding to an output vector from the $\alpha\text{-}\problem{cCVP}^p$ oracle.

\begin{claim}
\label{clm:simplenCVPreduction}
There is a polynomial-time, dimension-preserving reduction from \problem{CVP} to $\alpha\text{-}\problem{cCVP}^1$ for $\alpha(n) \leq C/\sqrt{n}$.
\end{claim}
\begin{proof}
On input $\lat \subset \R^n$ and $\vec{t} \in \R^n$, the reduction behaves as follows. First, if $n =1$, it solves the one-dimensional CVP instance in the straightforward way. Otherwise, it uses Theorem~\ref{thm:HKZtoCVP} and its cCVP oracle to compute a $(1+\alpha)$-HKZ basis $(\vec{b}_1, \ldots, \vec{b}_n)$ for $\lat$. It then calls its cCVP oracle on input $\lat$ and $\vec{t}$ and receives as output $\vec{y} \in \lat$. Let $(\gs{\vec{b}}_1, \ldots, \gs{\vec{b}}_n)$ be the Gram-Schmidt orthogonalization of the $\vec{b}_i$, and choose any index $k$ such that $ \length{\gs{\vec{b}}_k} > C\length{\vec{y} - \vec{t}}/\sqrt{n}$. Let $\lat' := \lat(\vec{b}_1,\ldots, \vec{b}_{k-1})$. The reduction then calls itself recursively on input $\lat'$ and $\vec{t} - \vec{y}$, receiving as output $\vec{x} \in \lat'$. Finally, it returns $\vec{y} + \vec{x}$.

It is clear that the reduction preserves dimension and runs in polynomial time. If $n = 1$, then correctness is also clear. Otherwise, by Lemma~\ref{lem:coveringradius}, 
$
\length{\vec{y} - \vec{t}}^2 < \sum \length{\gs{\vec{b}}_i}^2 \leq n \max \length{\gs{\vec{b}}_i}^2
$,
so there must exist an index $k$ as above. We assume for induction that the reduction is correct when the dimension of the lattice is less than $n$.  By the definition of cCVP, there is a vector $\vec{y}' \in \lat$ that is closest to $\vec{t}$ in $\lat$ with $\length{\vec{y} - \vec{y}'} \leq C \dist(\vec{t}, \lat)/\sqrt{n}$. Since $\length{\vec{y} - \vec{y}'} < \length{\gs{\vec{b}}_k} = \lambda_1(\lat')$, it follows that $\vec{y}' \in \lat' + \vec{y}$. 
By the induction hypothesis, $\vec{x}$ is a closest vector to $\vec{t} - \vec{y}$ in $\lat'$, and it follows that $\vec{y} + \vec{x}$ is a closest vector to $\vec{t}$ in $\lat' + \vec{y} = \lat' + \vec{y}'$. Therefore, $\vec{y} + \vec{x}$ is a closest vector to $\vec{t}$ in $\lat$, as needed.
\end{proof}

\subsection{Clusters of approximate closest lattice vectors}

\label{sec:clusters} 

We now wish to analyze a natural generalization of Claim~\ref{clm:simplenCVPreduction} that works with $\alpha\text{-}\problem{cCVP}^p$ for arbitrary $p$. In particular, we consider a reduction that solves \problem{CVP} recursively over many shifted sublattices $\lat' + \vec{y}_i$ where the $\vec{y}_i$ are the output of the cCVP oracle and $\lat'$ is some fixed sublattice. Correctness of such an algorithm follows immediately from Claim~\ref{clm:simplenCVPreduction}, but in order to bound the running time, we will need to bound the number of relevant shifts $\lat' + \vec{y}_i$. 

We accomplish this by showing that the approximate closest lattice vectors to $\vec{t}$ form \scarequotes{clusters} according to their cosets mod $2\lat$. This simple fact proves to be quite useful, and we will use it again in the next section to show that our DGS algorithm yields a solution to cCVP. We suspect that it will have other applications as well.

\begin{lemma}
\label{lem:clusters}
For any lattice $\lat \subset \R^n $, $\vec{t} \in \R^n$, $r_1,r_2 > 0$, and
$\vec{w}_1, \vec{w}_2 \in \lat-\vec{t}$ with $\vec{w}_1 \equiv \vec{w}_2
\imod{2\lat}$, if the $\vec{w}_i$ satisfy $\length{\vec{w}_i}^2 < \dist(\vec{t}, \lat)^2 + r_i^2$, then 
$\length{\vec{w}_1 - \vec{w}_2}^2 < 2(r_1^2+r_2^2)$.
\end{lemma}
\begin{proof}
Since $\vec{w}_1 \equiv \vec{w}_2 \imod{2\lat}$, we have that $(\vec{w}_1 +
\vec{w}_2)/2 \in \lat-\vec{t}$. Therefore, we have that
\begin{align*}
\|\vec{w}_1-\vec{w}_2\|^2 
   &= 2\|\vec{w}_1\|^2 + 2\|\vec{w}_2\|^2 - 4\|(\vec{w}_1+\vec{w}_2)/2\|^2 \\
   &< 2(\dist(\vec{t}, \lat)^2+r_1^2) + 2(\dist(\vec{t}, \lat)^2+r_2^2) - 4 \dist(\vec{t}, \lat)^2 \\
   &= 2(r_1^2+r_2^2) \; . \qedhere
\end{align*}
\end{proof}

In particular, Lemma~\ref{lem:clusters} shows that there are at most $2^n$ clusters of approximate closest points. We now derive an immediate corollary, which shows that, if the points are very close to $\vec{t}$, then each cluster lies in a shift of a lower-rank sublattice $\lat'$ defined in terms of a $\gamma$-HKZ basis, as we need for our reduction.

\begin{corollary}
\label{cor:sparse-project}
For any $\lat \subset \R^n$ with $\gamma$-HKZ basis $(\vec{b}_1, \ldots, \vec{b}_n)$ for some $\gamma \geq 1$, $\vec{t} \in
\R^n$, and $k \in [n]$, let $\lat' := \lat(\vec{b}_1,\ldots, \vec{b}_{k-1})$.
If $\vec{w}_1, \vec{w}_2 \in \lat - \vec{t}$ with $\vec{w}_1 \equiv \vec{w}_2
\imod{2\lat}$ satisfy $\length{\vec{w}_i}^2 < \dist(\vec{t}, \lat)^2 +
\length{\gs{\vec{b}}_k}^2/\gamma^2$, then $\vec{w}_1 \in \lat' + \vec{w}_2$.
\end{corollary}
\begin{proof}
Let $2\vec{v} = \vec{w}_1 - \vec{w}_2 \neq \vec0$. Note that $\vec{v} \in \lat$ by hypothesis, and by Lemma~\ref{lem:clusters}, we have that $\| \vec{v} \| < \length{\gs{\vec{b}}_k}/\gamma $. Since $\lambda_1(\pi_{\lat^{\prime \perp}}(\lat)) \geq \length{\gs{\vec{b}}_k}/\gamma$, it follows that $\vec{v} \in \lat'$, as needed. 
\end{proof}

To achieve our desired running time, we must show that, if $\lat'$ has relatively high rank, there must be significantly fewer than $2^n$ shifts of $\lat'$ that contain approximate closest vectors. This will allow us to bound the number of recursive calls that we make on high-rank sublattices. We accomplish this with the following two technical lemmas.

\begin{lemma} 
\label{lem:sparse-proj-remix}
For any $\lat \subset \R^n$ with $\gamma$-HKZ basis $(\vec{b}_1, \ldots, \vec{b}_n)$ for some $\gamma \geq 1$, $\vec{t} \in \R^n$, and $k \in [n]$, let $\lat' := \lat(\vec{b}_1,\ldots, \vec{b}_{k-1})$. If $r > 0$, $s>0$, and $k \leq \ell \leq n+1$ satisfy 
\begin{equation}
\label{eq:sparse-proj-cond}
r^2 + \frac{(k-1)}{2} \cdot \sum_{i=1}^{k-1}\|\gs{\vec{b}}_i\|^2 \leq \frac{1}{\gamma} \cdot \begin{cases} 
s^2\|\gs{\vec{b}}_k\|^2 &:  \ell = n+1 \\ 
\min \set{ s^2 \|\gs{\vec{b}}_k\|^2,
\|\gs{\vec{b}}_\ell\|^2 } &: \text{ otherwise } \end{cases} 
\; ,
\end{equation}
then we have that
\[
\big| \set{ \vec{c} \in \lat/\lat'\ :\ \dist(\vec{t}, \vec{c})^2 < \dist(\vec{t}, \lat)^2 + r^2} \big| \leq 2^{n-k+1}(2\ceil{2 s}^{\ell-k}-1) 
\; .
\]
\end{lemma}
\begin{proof}
For each $\vec{d} \in \lat/(2\lat + \lat')$, let 
\[
S_{\vec{d}} := \set{\vec{c} \in \lat/\lat'\ :\ \vec{c} \subset \vec{d} \text{ and } \dist(\vec{t}, \vec{c})^2 < \dist(\vec{t}, \lat)^2 + r^2}
\]
be the set of shifts of $\lat'$ that are subsets of $\vec{d}$ and contain an approximate closest vector. Since $\lat/\lat'$ is a refinement of $\lat/(2\lat + \lat')$ and $|\lat/(2\lat + \lat')| = 2^{n-k+1}$, it suffices to show that $|S_{\vec{d}}| \leq (2\ceil{2s}^{l-k}-1)$ for all $\vec{d} \in \lat/(2\lat + \lat')$.

Fix $\vec{d}$. Let $\vec{w}_1, \vec{w}_2 \in \vec{d} - \vec{t}$. Suppose $\length{\vec{w}_i}^2 < \dist(\vec{t}, \lat)^2 + r^2$.
A simple computation shows that there exist $a_1,\ldots, a_{k-1} \in \{-1,0,1\}$ such that $\vec{w}_1 - \sum_{i=1}^{k-1} a_i \vec{b}_i \equiv \vec{w}_2 \imod{2\lat}$ and
\[
\Big\| \vec{w}_1 - \sum_{i=1}^{k-1} a_i \vec{b}_i \Big\|^2 \leq \length{\vec{w}_1}^2 + \sum_{i=1}^{k-1} \length{\vec{b}_i}^2 < \dist(\vec{t}, \lat)^2 + r^2 + \sum_{i=1}^{k-1} \length{\vec{b}_i}^2 
\;.
\]
Since the $\vec{b}_i$ are $\gamma$-HKZ, we have that $\length{\vec{b}_i}^2 \leq \length{\gs{\vec{b}}_i}^2 + \frac{1}{4} \sum_{i=j}^{i-1}\length{\gs{\vec{b}}_j}^2$. Therefore, 
\[
\Big\| \vec{w}_1 - \sum_{i=1}^{k-1} a_i \vec{b}_i \Big\|^2 < \dist(\vec{t}, \lat)^2 + r^2  + (k-1)\sum_{i=1}^{k-1} \length{\gs{\vec{b}_i}}^2
\; .
\]
Let $2\vec{v} := \vec{w}_1 - \sum_{i=1}^{k-1} a_i \vec{b}_i - \vec{w}_2 \in 2\lat$. Since $\vec{w}_1 - \sum_{i=1}^{k-1} a_i \vec{b}_i \equiv \vec{w}_2 \imod{2\lat}$, we may apply Lemma~\ref{lem:clusters} to obtain
\[
\length{\vec{v}}^2 < r^2 + \frac{k-1}{2} \cdot \sum_{i=1}^{k-1} \length{\gs{\vec{b}_i}}^2 \; .
\]
Let $\pi_k := \pi_{\set{\vec{b}_1,\ldots, \vec{b}_{k-1}}^\perp}$ and $\M := \pi_k(\lat(\vec{b}_k,\ldots, \vec{b}_{\ell-1}))$. From the above, we have
\[
\length{\pi_k(\vec{v})}^2 \le \length{\vec{v}}^2 < r^2 + \frac{k-1}{2} \cdot \sum_{i=1}^{k-1} \length{\gs{\vec{b}_i}}^2
\; .
\]
Recalling the constraint on $\ell$ imposed by Eq.~\eqref{eq:sparse-proj-cond}, this implies that $\pi_k(\vec{v}) \in \M$. Furthermore, note that $\vec{w}_1 \in \lat' + \vec{w}_2$ if and only if $\pi_k(\vec{w}_1 - \vec{w}_2) = \pi_k(\vec{v}) = \vec0$. Therefore,
\[
|S_{\vec{d}}| \leq \Big| \Big\{ \vec{y} \in \M\ :\ \length{\vec{y}} < r^2 + \frac{k-1}{2} \cdot \sum_{i=1}^{k-1} \length{\gs{\vec{b}_i}}^2 \Big\} \Big|
\; .
\]

Finally, note that $\lambda_1(\M) \geq \length{\gs{\vec{b}}_k}/\gamma$. By Eq.~\eqref{eq:sparse-proj-cond} the length bound in the above equation is at most $s \lambda_1(\M)$. The result then follows from applying Theorem~\ref{thm:lat-pt-bnd} and noting that $\dim \M = \ell - k$.
\end{proof}

This next lemma shows that we can choose an index $k$ such that either $\lat'$ has fairly small rank or relatively few shifts of $\lat'$ contain approximate closest vectors.

\begin{lemma}
\label{lem:good-index}
For any lattice $\lat \subset \R^n$ with $\gamma$-HKZ basis $(\vec{b}_1,\ldots, \vec{b}_n)$ for some $n \geq 2$ and $1 \leq \gamma \leq 1+\frac{1}{10n^2}$, any efficiently computable function $f : \Z^+ \mapsto \Z^+$, and
 \[
 r := n^{-2f(n)} \max_{i \in [n]} \|\gs{\vec{b}}_i\| %
\; , 
\]
there exists $k \in [n]$ such that
 if $\lat' := \lat(\vec{b}_1,\ldots, \vec{b}_{k-1})$, then $\length{\gs{\vec{b}}_k} \geq \gamma \cdot \frac{\mu(\lat)}{n^{2f(n)}}$ and
\[
\big| \set{ \vec{c} \in \lat/\lat'\ :\ \dist(\vec{t}, \vec{c})^2 <
\dist(\vec{t}, \lat)^2 + r^2} \big| \leq \begin{cases} 2^{n-k+1}&:  \text{ if }\; n-f(n) < k
\leq n \\ 2^{n-k+2}n^{n/f(n)}&: \text{ otherwise} \end{cases}
\]
Furthermore, the index $k$ can be computed efficiently from the $\vec{b}_i$.
\end{lemma}
\begin{proof}
Let $R := \max_{i \in [n]} \|\gs{\vec{b}}_i\| = n^{2f(n)} r$. Define $m_j \in [n]$ for $0 \le j < 2f(n)$ to be the smallest index $i$ such that $ \|\gs{\vec{b}}_{i}\| \ge \gamma \cdot \frac{R}{n^j}$. Then, by definition, we have that $m_0 \ge m_1 \ge \cdots \ge m_{2f(n)-1}$.
Furthermore,
\begin{align}
\label{eq:gap}
r^2 + \frac{m_j-1}{2} \cdot \sum_{i=1}^{m_j-1} \|\gs{\vec{b}}_i\|^2 & < R^2 \cdot \Big(\frac{1}{n^{4f(n)}} + \gamma^2 \cdot \frac{(m_j-1)^2}{n^{2j}} \Big) \nonumber \\  
 & \leq \frac{R^2}{n^{2j}} \cdot \Big(\frac{1}{n^{4f(n)-2j}} + \gamma^2 \cdot (n-1)^2 \Big) \nonumber \\
& < \frac{ R^2}{n^{2j-2}} \;.                                                              
\end{align}
First, consider the case when there exists $j \leq f(n)$ such that $m_{j} = m_{j-1}$. In this case, we claim that the required index is $k = m_j$. To see this, simply note that $\|\gs{\vec{b}}_k\| \ge \gamma \cdot \frac{R}{n^{j-1}}$ by definition. Then, by Eq.~(\ref{eq:gap}), the conditions for Lemma~\ref{lem:sparse-proj-remix} are satisfied with $\ell = k$ and $s = n$.  Applying Lemma~\ref{lem:coveringradius} gives $\|\gs{\vec{b}}_k\| > \gamma \frac{\mu(\lat)}{n^{2f(n)}}$, as needed.

So, it suffices to assume that  $m_0 > m_1 > \cdots > m_{f(n)}$. In this case, clearly $m_{f(n)} \le n - f(n)$. Now, by the pigeonhole principle, there exists $j \in \{f(n), f(n) + 1, \ldots, 2f(n)-1\}$ such that $m_{j-1} - m_j< \frac{n}{f(n)}$. Then, let $k = m_j$, and $\ell = m_{j-1}$. Noting the fact that $\|\gs{\vec{b}}_k\| \ge \frac{R}{n^{j}}$ and $\|\gs{\vec{b}}_\ell\| \ge \frac{R}{n^{j-1}}$, the bound on the number of shifts follows from Lemma~\ref{lem:sparse-proj-remix} and Eq.~(\ref{eq:gap}). The bound on $\|\gs{\vec{b}}_k\|$ again follows from applying Lemma~\ref{lem:coveringradius}.
\end{proof}

\subsection{The reduction}

We can now present our more general reduction. We note in passing that, if the cCVP oracle happens to output a nearby point for each exact closest lattice vector, then (a minor modification of) our reduction actually finds all closest vectors.

\begin{theorem}
\label{thm:CVPtonCVP}
For any constant $\delta \in [0,1)$, there is a reduction from exact CVP to $\alpha\text{-}\problem{cCVP}^p$ where $\alpha(n) :=  1/(10n^{4n^{\delta} +1})$ such that the maximal number of oracle calls that the reduction makes on lattices of dimension $d$ when the input lattice has dimension $n$ is 
\[
g(n,d) \leq \min\Big\{ 2^{n-d+O(n^{2-2\delta} \log n)},\ \poly(n) \prod_{i=d+1}^n p(i) \Big\} 
\; .
\] The running time of the reduction is $\poly(n) \cdot \sum_d p(d) g(n,d)$.
\end{theorem}
\begin{proof}
The reduction behaves quite similarly to the simple procedure from Claim~\ref{clm:simplenCVPreduction}. The only difference is that this new reduction chooses $\lat'$ more carefully and makes recursive calls on many shifts of $\lat'$ corresponding to the many outputs of its $\alpha\text{-}\problem{cCVP}^p$ oracle. In particular, on input $\lat \subset \R^n$ and $\vec{t} \in \R^n$, the reduction behaves as follows. First, if $n = 1$, it solves the one-dimensional CVP instance in the straightforward manner. Otherwise, it uses Theorem~\ref{thm:HKZtoCVP} and its oracle to compute a $(1+\alpha)$-HKZ basis $(\vec{b}_1, \ldots, \vec{b}_n)$ for $\lat$. It then calls its oracle on input $\lat$ and $\vec{t}$ and receives as output $\vec{y}_1,\ldots, \vec{y}_{\hat{p}} \in \lat$. As we noted below Definition~\ref{def:nCVP-new}, we may assume without loss of generality that
\begin{equation}
\label{eq:additivedistancebound}
\length{\vec{y}_j - \vec{t}}^2 
\leq (1+\alpha(n))^4 \dist(\vec{t}, \lat)^2 < \dist(\vec{t}, \lat)^2 + n^{-4n^\delta} \max_{i \in [n]} \|\gs{\vec{b}}_i\|^2 \; .
\end{equation}

The reduction then computes the index $k$ as in Lemma~\ref{lem:good-index} with $f(n) := n^\delta$. Let $\lat' := \lat(\vec{b}_1,\ldots, \vec{b}_{k-1})$. The reduction groups the $\vec{y}_i$ according to their coset mod $\lat'$. For each such coset $\vec{c}$, it picks an arbitrary representative $\vec{y}_{\vec{c}} \in \vec{c}$ and calls itself recursively on input $\lat'$ and $\vec{t} - \vec{y}_{\vec{c}}$, receiving as output $\vec{x}_{\vec{c}}$. Finally, it outputs the closest $\vec{x}_{\vec{c}} + \vec{y}_{\vec{c}}$ to $\vec{t}$.

Correctness follows immediately from the proof of Claim~\ref{clm:simplenCVPreduction}. In particular, consider a sequence of recursive calls such that the corresponding $\vec{y}_{\vec{c}}$ represent valid solutions to their respective $\alpha\text{-}\problem{cCVP}^1$ instances and note that the reduction behaves identically to the procedure from Claim~\ref{clm:simplenCVPreduction} along this sequence.

The statement about the running time is clear. We now analyze the number of recursive calls. Consider a single thread with $\dim \lat = n$ and $\dim \lat' = \hat{n}$. The total number of recursive calls made by this thread is
\begin{align}
L(n, \hat{n}) &:= \big| \set{ \vec{c} \in \lat/\lat'\ :\ \exists\ i \text{ with } \vec{y}_{i} \in \vec{c}} \big| \nonumber \\
&\leq 
\min \Big\{ p(n)\ ,\ 
\big| \set{ \vec{c} \in \lat/\lat'\ :\ \dist(\vec{t}, \vec{c})^2 <\dist(\vec{t}, \lat)^2 + n^{-2f(n)}} \big| 
\Big\}
\label{eq:Lbound}
\; .
\end{align}
Note that $g(n,d)$ satisfies the recurrence relation
\begin{equation}
g(n, d) \le \max_{d \leq \hat{n} < n} L(n, \hat{n}) g(\hat{n},d) 
\; ,
\label{eq:grecurrence}
\end{equation}
with base case $g(d,d) = \poly(n)$.
The bound $g(n,d) \leq \poly(n) \prod_{i=d+1}^n p(i)$ follows immediately from the fact that $L(n,\hat{n}) \leq p(n)$.

Now, we wish to prove by induction that for any $d$ and $n$, we have $g(n,d) \leq 2^{n-d+C^*n^{2 - 2\delta}\log n}$ for some constant $C^*$. For $n=1$ or $d=n$, this is trivial. Suppose that the induction hypothesis holds for dimensions less than $n$. 
By Eq.~\eqref{eq:grecurrence}, it suffices to prove that $L(n, \hat{n}) g(\hat{n},d) \leq 2^{n-d+C^*n^{2 - 2\delta}\log n}$ for all $\hat{n} < n$. Note that Eq.~\eqref{eq:additivedistancebound} gives us the bound that we need to apply Lemma~\ref{lem:good-index}. Plugging the lemma into Eq.~\eqref{eq:Lbound}, we have 
\[
L(n, \hat{n}) \leq 
\begin{cases} 2^{n-\hat{n}}
&:  \text{ if }\; \hat{n} \geq n-f(n)  \\ 
2^{n-\hat{n}+1}n^{n/f(n)}
&: \text{ otherwise} \end{cases}
\; .
\]
If $\hat{n} \geq n-f(n)$, then by this bound and the induction hypothesis,
\[
L(n, \hat{n}) g(\hat{n},d)  \leq 2^{n-\hat{n}} \cdot g(\hat{n},d) \leq 2^{n-d+C^* n^{2-2\delta} \log n}
\; ,
\]
as needed. Otherwise, $\hat{n} < n-f(n)$, and we have
\begin{align*}
L(n, \hat{n}) g(\hat{n},d) 
&\leq 2^{n-d+1 + C^* \hat{n}^{2-2\delta}\log n}n^{n/f(n)}\\
&\leq 2^{n-d+1 + C^* (n-f(n))^{2-2\delta}\log n + n\log_2 n/f(n)}\\
&\leq 2^{n-d+1 + C^* n^{2-2\delta} - C^*(2-2\delta) n^{1-2\delta}f(n)\log n + n\log_2 n/f(n)}\\
&\leq 2^{n-d + C^* n^{2-2\delta}}
\; ,
\end{align*}
as needed.
\end{proof}
\section{Finishing the proof}
\label{sec:DGSsolvesnCVP}
\subsection{The mass of cosets with closest vectors}

We now show that our DGS algorithm yields a solution to $\alpha\text{-}\problem{cCVP}^p$, i.e., that one of its output vectors will be very close to an exact shortest vector in the shifted lattice with high probability when called with appropriate parameters. (See Definition~\ref{def:nCVP-new}.) By our \scarequotes{cluster} analysis in Section~\ref{sec:clusters}, this reduces to showing that one of the output vectors will be a short vector that is in the same coset of $2\lat$ as a shortest vector. Since the number of samples returned by our algorithm is essentially the number that we need to \scarequotes{see each coset with relatively high Gaussian mass,} it would suffice to show that any coset of $2\lat - \vec{t}$ that contains a shortest vector must have high mass. Instead, we are only able to prove the slightly weaker (but still sufficient) fact that for a suitable list of parameters $s_1,\ldots, s_\ell$, each such coset has high mass with respect to the discrete Gaussian with at least one of these parameters. (See Corollary~\ref{cor:bigcoset}.)

\begin{lemma}
\label{lem:bigcoset1}
Let $\lat \subset \R^n$ be a lattice and $\vec{t} \in \R^n$ with $\vec{y} \in \lat$ a closest vector to $\vec{t}$ in $\lat$. Then, for any  $s > 0$,
\[ 
1 \le \frac{ \max_{\vec{c} \in \lat/(2\lat)}\rho_{s}(\vec{c} - \vec{t})}{\rho_{s}(\vec{y} - \vec{t}) \cdot \rho_{s}(2\lat)} \le \frac{\prod_{j=1}^{\infty} \rho_{2^{-j/2}s}(\lat)^{1/2^j}}{\rho_{s}(2\lat)} \le 2^{n/4} 
\; .
\]
\end{lemma}
\begin{proof}
The first inequality trivially follows from Lemma~\ref{lem:betterrhoLtbound}. Let $\theta(i) := \rho_{2^{-i/2}s}(\lat)$ and $\phi(i) := \max_{\vec{c} \in \lat/(2\lat)}\rho_{2^{-i/2}s}(\vec{c} - \vec{t})$.
By Corollary~\ref{cor:RSHolder}, we have
\[
\phi(i) \leq \phi(i+1)^{1/2} \theta(i+1)^{1/2}\;.\]
 Applying this inequality $k$ times, we have
\[
\phi(0) \leq \phi(k)^{1/2^k} \cdot \prod_{j=1}^{k} \theta(j)^{1/2^j}
\; .
\]
We take the limit as $k \rightarrow \infty$. Since $\vec{y} \in \lat$ is a closest vector to $\vec{t}$, we have  
\[
\lim_{k \rightarrow \infty} \phi(k)^{1/2^k} = \rho_{s}(\vec{y} - \vec{t}) \;.
\]
The second inequality is then immediate. 
 For the third inequality, note that
for all $i \geq 2$, $\theta(i) \leq \theta(2) = \rho_s(2\lat)$, and by Lemma~\ref{lem:banaszczyk}, $\theta(1) \leq 2^{n/2} \theta(2)$. Therefore, 
\[
\prod_{j=1}^\infty \theta(j)^{1/2^j} \leq 2^{n/4} \cdot \prod_{j=1}^\infty \theta(2)^{1/2^j} = 2^{n/4} \cdot \theta(2)
\; . \qedhere
\]
\end{proof}

We will need the following technical lemma to obtain a stronger version of Lemma~\ref{lem:bigcoset1}.

\begin{lemma}
\label{lem:bigcoset2}
For any lattice $\lat \subset \R^n$, $s > 0$, and integer $\ell >0$, there exists an integer $1 \le i \le \ell$ such that 
\[
\frac{\prod_{j=1}^{\infty} \rho_{2^{-(i+j)/2}s}(\lat)^{1/2^j}}{\rho_{2^{-i/2}s}(2\lat)}  \le 2^{\frac{3n}{4\ell}}
\; .
\]
\end{lemma}
\begin{proof}
For $i \ge 0$, let $\theta(i) := \rho_{2^{-i/2}s}(\lat)$ as in the previous proof.
Let
\[
S_i := \frac{\prod_{j=1}^{\infty} \theta(i+j)^{1/2^j}}{\theta(i+2)} 
\; ,
\]
and 
\[
R_i := \frac{\theta(i+1)}{\theta(i+2)} \;.
\]
We need to show that there exists an integer $1 \le i \le \ell $ such that $S_i \le 2^{3n/4\ell}$.

By Lemma~\ref{lem:bigcoset1}, we have that for all $i$, $1 \le S_i \le 2^{n/4}$, and by Lemma~\ref{lem:banaszczyk}, we have that, $1 \le R_i \le 2^{n/2}$. 
Note that
\[
\frac{S_i^2}{S_{i+1}} = \frac{\theta(i+1) \cdot \theta(i+3)}{\theta(i+2)^2} = \frac{R_i}{R_{i+1}} \;.
\]
Therefore,
\begin{equation*}
	2^{n/2}										\ge \frac{R_0}{R_{\ell+1}} 
											= \prod_{i = 0}^{\ell} \frac{R_i}{R_{i+1}} 
                       = \prod_{i = 0}^\ell \frac{S_i^2}{S_{i+1}} 
                       = \frac{S_0^2}{S_{\ell+1}} \prod_{i=1}^\ell S_i 
											 \ge \frac{1}{2^{n/4}} \prod_{i=1}^\ell S_i \;,
\end{equation*}
where the first inequality uses $R_0 \le 2^{n/2}$ and $R_{\ell + 1} \ge 1$, and the last inequality uses $S_0 \ge 1$ and $S_{\ell + 1} \le 2^{n/4}$. The result then follows.
\end{proof}

Finally, we have the following corollary, which follows immediately from Lemmas~\ref{lem:bigcoset1} and~\ref{lem:bigcoset2}, and Lemma~\ref{lem:betterrhoLtbound}. The corollary shows that, if $\vec{c} \in \lat/(2\lat)$ contains a closest vector to $\vec{t}$ and we sample from $D_{\lat - \vec{t}, s}$ for many different values of $s$, then $\vec{c} - \vec{t}$ will have relatively high weight for at least one parameter $s$.

\begin{corollary}
\label{cor:bigcoset}
For any lattice $\lat \subset \R^n$ and $\vec{t} \in \R^n$, let $\vec{y} \in \lat$ a closest vector to $\vec{t}$ in $\lat$. Then, for any $s > 0$ and integer $\ell > 0$, there exists an integer $1 \le i \le \ell$ such that 
\[
1 
\leq \frac{ \max_{\vec{c} \in \lat/(2\lat)}\rho_{2^{-i/2}s}(\vec{c} - \vec{t})}{\rho_{2^{-i/2}s}(2\lat + \vec{y} - \vec{t})} 
\le \frac{ \max_{\vec{c} \in \lat/(2\lat)}\rho_{2^{-i/2}s}(\vec{c} - \vec{t})}{\rho_{2^{-i/2}s}(\vec{y} - \vec{t}) \cdot \rho_{2^{-i/2}s}(2\lat)} 
\le 2^{\frac{3n}{4\ell}} 
\; .
\]
\end{corollary}

\subsection{The cCVP algorithm}

With Corollary~\ref{cor:bigcoset}, it is almost immediate that the algorithm from Theorem~\ref{thm:DGS} yields a solution to cCVP. Below, we make this formal.

\begin{theorem}
\label{thm:nCVP}
For any efficiently computable function $f(n) \geq n^{\omega(1)}$, there is an algorithm that solves $(1/f(n))\text{-}\problem{cCVP}^{p}$ with probability at least $1-2^{-Cn^2}$ in time $2^{n+O(\log n \log f(n)+n/\log f(n))}$, where $p(n) := \poly(n) \cdot 2^{n+O(n/\log f(n))}  $.
\end{theorem}
\begin{proof}
On input a lattice $\lat \subset \R^n$ and shift $\vec{t} \in \R^n$, the algorithm first calls the procedure from Corollary~\ref{cor:approxCVP} to compute $\tilde{d}$ with $\dist(\vec{t}, \lat)/2 \leq \tilde{d} \leq \dist(\vec{t}, \lat)$. Let $s := \tilde{d}/(n^3 f(n)) $. For $i = 0, \ldots, \ell := \ceil{\log 10 f(n)}$, the algorithm runs the procedure from Theorem~\ref{thm:DGS} $n^2 \cdot \ceil{ 2^{n/\ell}}$ times with input $\lat$, $\vec{t}$, and $s_i := 2^{-i/2}s$, receiving as output a total of $\hat{m}_i \geq n^2 2^{n/\ell} \cdot  m(\lat - \vec{t}, s_i)$ vectors $(\vec{X}_{i,1},\ldots, \vec{X}_{i,\hat{m}_i}) \in  \lat - \vec{t}$. (We may assume that $\hat{m}_i \leq n^2 2^{n} \cdot \ceil{2^{n/\ell}}$, since we can trivially truncate the output of each run at $2^n \geq m(\lat - \vec{t}, s_i)$ vectors.) For each $i,j$, let $\vec{y}_{i,j} := \vec{X}_{i,j} + \vec{t} \in \lat$. Finally, the algorithm outputs the $\vec{y}_{i,j}$.

The running time is dominated by the running time of the $\ell n^2 2^{n/\ell}$ applications of Theorem~\ref{thm:DGS}. So, the algorithm runs in time $\ell n^2 2^{n + O(\log n\log f(n)) + n/\ell} = 2^{n+O(\log n \log f(n)+n/\log f(n))}$. The value for $p(n)$ follows from the assumed bound on $\hat{m}_i$. 

To prove correctness, first note that by Theorem~\ref{thm:DGS}, up to statistical distance $2^{-Cn^2}$, we may assume that the $\vec{X}_{i,j}$ are distributed exactly as independent discrete Gaussians $D_{\lat - \vec{t}, s_i}$. Then, by Corollary~\ref{cor:banaszczykcor}, all of the output vectors are $(1+1/f(n))$-approximate closest vectors except with probability at most $2^{-Cn^2}$. So, by Corollary~\ref{cor:sparse-project}, it suffices to show that with high probability there is some $i,j$ such that $\vec{y}_{i,j}$ is in the same coset mod $2\lat$ as a closest vector $\bar{\vec{y}} \in \lat$ to $\vec{t}$. Fix $i$ as in Corollary~\ref{cor:bigcoset}. Then, for any $j$,
\begin{align*}
\Pr[\vec{y}_{i,j} \equiv \bar{\vec{y}} \imod{2\lat}] &= \frac{ \rho_{s_i}(2\lat + \bar{\vec{y}} - \vec{t})}{\rho_{s_i}(\lat - \vec{t})}  \\
&= \frac{1}{m(\lat - \vec{t}, s_i)} \cdot \frac{ \rho_{s_i}(2\lat + \bar{\vec{y}} - \vec{t})}{\max_{\vec{c} \in \lat/(2\lat)}\rho_{s_i}(\vec{c} - \vec{t})}\\
&\geq \frac{n^2 2^{n/\ell}}{\hat{m}_i}\cdot 2^{-\frac{3n}{4\ell}} &\text{(Corollary~\ref{cor:bigcoset})}\\
&> \frac{2n^2}{\hat{m}_i}
\;.
\end{align*}
The result follows by recalling that the $\vec{y}_{i,j}$ are independent.
\end{proof}

We obtain our main result as a corollary. We note in passing that a simple union bound shows that the algorithm from Theorem~\ref{thm:nCVP} actually finds a nearby vector for \emph{each} closest lattice vector. Together with the remark above Theorem~\ref{thm:CVPtonCVP}, this shows that we can actually find \emph{all} closest vectors in time $2^{n+o(n)}$.

\begin{corollary}
\label{cor:CVP}
There is an algorithm that solves \emph{exact} CVP (with high probability) in time $2^{n+O(n^{2/3}\log^2 n)}$.
\end{corollary}
\begin{proof}
Combine the algorithm from Theorem~\ref{thm:nCVP} with $f(n) := 2^{2n^{2/3} \log n}$ with the reduction from Theorem~\ref{thm:CVPtonCVP} with $\delta := 2/3$. (By applying a union bound over all oracle calls in the reduction, we see that the error is not an issue.)
\end{proof}

\section*{Acknowledgments } 
We thank Oded Regev and Alexander Golovnev for many helpful conversations. We also thank the anonymous FOCS reviewers for their valuable comments.

\bibliographystyle{alpha}
\newcommand{\etalchar}[1]{$^{#1}$}

\appendix
\section{Proof of Lemma~\ref{lem:sumofgaussians}}

\begin{proof}[Proof of Lemma~\ref{lem:sumofgaussians}]
Multiplying the left-hand side of~\eqref{eq:sumofgaussians} by 
$\Pr_{(\vec{X}_1, \vec{X}_2) \sim D_{\lat-\vec{t}, s}^2}[\vec{X}_1 + \vec{X}_2 \in 2\lat - 2 \vec{t}]$, 
we get for any $\vec{y} \in \lat - \vec{t}$,
\begin{align*}
\Pr_{(\vec{X}_1, \vec{X}_2) \sim D_{\lat-\vec{t}, s}^2}[(\vec{X}_1 + \vec{X}_2)/2 = \vec{y}]  
&= \frac{1}{\rho_s(\lat-\vec{t})^2}\cdot\sum_{\vec{x} \in \lat-\vec{t}} \rho_s(\vec{x}) \rho_s(2\vec{y} - \vec{x})\\
&= \frac{\rho_{s/\sqrt{2}}(\vec{y})}{\rho_s(\lat-\vec{t})^2}\cdot\sum_{\vec{x} \in \lat-\vec{t}} \rho_{s/\sqrt{2}}(\vec{x} - \vec{y})\\
&= \frac{\rho_{s/\sqrt{2}}(\vec{y})}{\rho_s(\lat-\vec{t})^2}\cdot \rho_{s/\sqrt{2}}(\lat) \; . 
\end{align*}
Hence both sides of~\eqref{eq:sumofgaussians} are proportional to each other. Since they are probabilities, they are actually equal.
\end{proof}

\section{Proof of Proposition~\ref{prop:combiner}}

\begin{proof}[Proof of Proposition~\ref{prop:combiner}]
Let $(\vec{X}_1, \ldots, \vec{X}_M)$ be the input vectors. 
For each $i$, let $\coset_i \in \lat/(2\lat)$ be such that $\vec{X}_i \in \coset_i - \vec{t}$. The combiner runs the algorithm from Theorem~\ref{thm:squaresampler} with input $\kappa$
and $(\coset_1,\ldots,\coset_M)$,
receiving output $(\coset_1', \ldots, \coset_m')$. 
(Formally, we must encode the cosets as integers in $\{1,\ldots, 2^n\}$.) Finally, for each $\coset_i'$, it chooses a pair of unpaired vectors $\vec{X}_j, \vec{X}_k$ with 
$\coset_j = \coset_k = \coset_i'$ and outputs $\vec{Y}_i = (\vec{X}_j + \vec{X}_k)/2$.

The running time of the algorithm follows from Item~\ref{item:squareruntime} of Theorem~\ref{thm:squaresampler}. Furthermore, we note that by Item~\ref{item:squareinputoutput} of the same theorem, there will always be a pair of indices $j,k$ for each $i$ as above.

To prove correctness, we observe that for $\coset \in \lat/(2\lat)$ and $\vec{y} \in \coset - \vec{t}$, 
\[ 
\Pr[\vec{X}_i = \vec{y}] = \frac{\rho_s(\coset - \vec{t})}{\rho_s(\lat - \vec{t})} \cdot \Pr_{\vec{X} \sim D_{\coset - \vec{t},s}}[\vec{X} = \vec{y}] \;.
\]
In particular, we have that $\Pr[\coset_i = \coset] = \rho_s(\coset-\vec{t})/\rho_s(\lat - \vec{t})$. 
Then, the cosets $(\coset_1,\ldots,\coset_M)$ satisfy the conditions necessary for Item~\ref{item:squaredistribution} of Theorem~\ref{thm:squaresampler}. 

Applying the theorem, up to statistical distance $M \exp(C_1 n - C_2 \kappa)$, we have that the output vectors are independent, and
\begin{align*}
m 
&\geq M \cdot \frac{1}{32\kappa }\cdot\frac{\sum_{\coset \in \lat/(2\lat)}\rho_s(\coset-\vec{t})^2}{\rho_s(\lat-\vec{t})\max_{\coset \in \lat / (2\lat)} \rho_s(\coset - \vec{t})} \\
&= M \cdot \frac{1}{32\kappa}\cdot\frac{\rho_{s/\sqrt{2}}(\lat) \cdot \rho_{s/\sqrt{2}}(\lat - \vec{t})}{ \rho_s(\lat-\vec{t})\max_{\coset \in \lat / (2\lat)} \rho_s(\coset - \vec{t})}
\; ,
\end{align*}
where the equality follows from Lemma~\ref{lem:RS15} by setting $\vec{x} = \vec{t}$, and $\vec{y} = \vec{0}$.
 Furthermore, we have $\Pr[\coset_i' = \coset] = \rho_s(\coset-\vec{t})^2/\sum_{\coset'} \rho_s(\coset' - \vec{t})^2$ for any coset $\coset \in \lat/(2\lat)$. Therefore, for any $\vec{y} \in \lat$, 
\begin{align*}
\Pr[\vec{Y}_i = \vec{y}] &= \frac{1}{\sum \rho_s(\coset - \vec{t})^2} \cdot \sum_{\coset \in \lat/(2\lat)} \rho_s(\coset-\vec{t})^2 \cdot \Pr_{(\vec{X}_j, \vec{X}_k) \sim D_{\coset-\vec{t}, s}^2}[(\vec{X}_j + \vec{X}_k)/2  = \vec{y}]\\
&= \Pr_{(\vec{X}_1, \vec{X}_2) \sim D_{\lat-\vec{t}, s}^2}[(\vec{X}_1 + \vec{X}_2)/2 = \vec{y} ~|~ \vec{X}_1 + \vec{X}_2 \in 2\lat - 2\vec{t}] 
\; .
\end{align*}
The result then follows from Lemma~\ref{lem:sumofgaussians}.
\end{proof}

\end{document}